\documentclass[journal]{IEEEtran}

\usepackage{xcolor}
\usepackage{cite}
\usepackage{graphicx}
\usepackage{amsmath} 
\usepackage{amssymb}  
\usepackage{amsfonts}
\usepackage{mathrsfs}
\usepackage{mathtools}
\usepackage{xpatch}
\usepackage{tablefootnote}
\usepackage{pifont}
\usepackage[bbgreekl]{mathbbol}
\usepackage{multirow}
\usepackage{makecell}
\usepackage{booktabs}
\usepackage[caption=false,font=footnotesize,subrefformat=parens,labelformat=parens]{subfig}
\usepackage{enumitem}
\usepackage{amsthm} 
\usepackage{tikz} 
\usepackage{booktabs,threeparttable,xcolor}
\newcommand{\Plus}{\mathord{\text{\ding{58}}}}
\newcommand{\Xmark}{\mathord{\text{\ding{54}}}}

\definecolor{mycolor1}{rgb}{0.00000,0.44700,0.74100}%
\definecolor{mycolor2}{rgb}{0.85000,0.32500,0.09800}%
\definecolor{mycolor3}{rgb}{0.92900,0.69400,0.12500}%
\definecolor{mycolor4}{rgb}{0.49400,0.18400,0.55600}%
\definecolor{mycolor5}{rgb}{0.46600,0.67400,0.18800}%
\definecolor{mycolor6}{rgb}{0.30100,0.74500,0.93300}%

\newcommand{\blueline}{\raisebox{2pt}{\tikz{\draw[-,mycolor1,solid,line width = 1pt](0,0) -- (3mm,0);}}}
\newcommand{\dashedblueline}{\raisebox{2pt}{\tikz{\draw[-,mycolor1,dashed,line width = 1pt](0,0) -- (3mm,0);}}}
\newcommand{\redmark}{\tikz[baseline=-0.25em]\node[mycolor2,line width = 1pt, mark size=2.5pt]{\pgfuseplotmark{x}};}
\newcommand{\redline}{\raisebox{2pt}{\tikz{\draw[-,mycolor2,solid,line width = 1pt](0,0) -- (3mm,0);}}}

\newcommand{\yellowline}{\raisebox{2pt}{\tikz{\draw[-,mycolor3,solid,line width = 1pt](0,0) -- (3mm,0);}}}
\newcommand{\blackline}{\raisebox{2pt}{\tikz{\draw[-,black,solid,line width = 1pt](0,0) -- (3mm,0);}}}
\newcommand{\greyd}{\raisebox{2pt}{\tikz{\draw[-,black!50!white,dashed,line width = 1pt](0,0) -- (3mm,0);}}}
\newcommand{\greydd}{\raisebox{2pt}{\tikz{\draw[-,black!50!white,dash dot,line width = 1pt](0,0) -- (4mm,0);}}}

\newcommand{\purpleline}{\raisebox{2pt}{\tikz{\draw[-,mycolor4,solid,line width = 1pt](0,0) -- (3mm,0);}}}
\newcommand{\greenline}{\raisebox{2pt}{\tikz{\draw[-,mycolor5,solid,line width = 1pt](0,0) -- (3mm,0);}}}
\newcommand{\cyanline}{\raisebox{2pt}{\tikz{\draw[-,mycolor6,solid,line width = 1pt](0,0) -- (3mm,0);}}}

\DeclareRobustCommand{\Wyegnd}{\mathbin{
\tikz[x=1pt, y=1pt, scale=0.85]{\draw
    (-1.4, 0) -- (1.4, 0)
    (-1, -1) -- (1, -1)
    (-0.4, -2) -- (0.4, -2)
    (0, 0) -- ++(0, 2) -- ++(-3, 0) coordinate (tmp)
    -- +(0, -4)
    (tmp) +(45:4) -- (tmp) -- +(135:4)
    ;}}}

\DeclareRobustCommand{\wyegnd}{\mathbin{
\tikz[x=1pt, y=1pt, scale=0.7]{\draw
    (-1.4, 0) -- (1.4, 0)
    (-1, -1) -- (1, -1)
    (-0.4, -2) -- (0.4, -2)
    (0, 0) -- ++(0, 2) -- ++(-3, 0) coordinate (tmp)
    -- +(0, -4)
    (tmp) +(45:4) -- (tmp) -- +(135:4)
    ;}}}

\makeatletter
\newcolumntype{"}{@{\hskip\tabcolsep\vrule width 1pt\hskip\tabcolsep}}
\makeatother

\makeatletter
\xpatchcmd{\@thm}{\thm@headpunct{.}}{\thm@headpunct{}}{}{}
\makeatother

\DeclareSymbolFontAlphabet{\mathbb}{AMSb}
\DeclareSymbolFontAlphabet{\mathbbl}{bbold}

\DeclareMathOperator{\diag}{diag}

\DeclareMathOperator{\diff}{d}

\DeclareMathOperator*{\argmin}{argmin}

\makeatletter
\renewcommand*\env@matrix[1][\arraystretch]{%
  \edef\arraystretch{#1}%
  \hskip -\arraycolsep
  \let\@ifnextchar\new@ifnextchar
  \array{*\c@MaxMatrixCols c}}
\makeatother

\newcommand{\R}{{\mathbb R}}

\newcommand{\mc}{\mathcal}
\newcommand{\ddt}{\tfrac{\diff}{\diff \!t}}
\newcommand{\norm}[1]{\left \lVert #1 \right \rVert}

\newtheorem{proposition}{Proposition}

\newtheorem{remark}{Remark}

\begin{document}

\title{Constraint-Aware Grid-Forming Control for Current Limiting}

\author{Dominic Gro\ss{},~\IEEEmembership{Senior Member,~IEEE}\thanks{
  
This material is based upon work supported by the U.S. Department of Energy's Office of Energy Efficiency and Renewable Energy (EERE) under the Solar Energy Technologies Office and the National Science Foundation under Grant 2143188. The views expressed herein do not necessarily represent the views of the U.S. Department of Energy or the United States Government.
D. Gro\ss{} is with the Department of Electrical and Computer Engineering, University of Wisconsin-Madison, Madison, WI 53706, USA (e-mail: dominic.gross@wisc.edu)}}

\maketitle

\begin{abstract}
This work develops a constraint-aware grid-forming (GFM) control that explicitly accounts for current limits and modulation limits within the GFM oscillator dynamics generating the GFM voltage reference (i.e., phase angle and magnitude). Broadly speaking, the voltage reference generated by the constraint-aware GFM control minimizes the deviation from conventional unconstrained GFM droop control, while respecting current and modulation limits. The resulting GFM control achieves fast current limiting while preserving transient stability, e.g., exhibiting infinite critical clearing time. To develop the control, we first characterize and analyze the set of converter voltages that do not result in constraint violations. Next, an efficient algorithm for projecting voltages onto the feasible set is developed. Subsequently, these results are used to restrict the dynamics of GFM droop control to the set of feasible voltages. Finally, detailed simulation studies and hardware experiments are used to illustrate and validate the response to short-circuit faults and phase jumps.
\end{abstract}

\begin{IEEEkeywords}
  Grid-forming control, current limiting, fault ride through
\end{IEEEkeywords}

\section{Introduction}
\IEEEPARstart{P}{ower} electronics have become ubiquitous in electric power systems, interfacing renewable generation, energy storage systems, high voltage direct current (HVDC) transmission, and a wide range of industrial
and domestic loads. This large-scale integration of converter-interfaced
resources results in significantly different system dynamics
that challenge today's operating and control paradigms~\cite{GD2022}.

The majority of converter-interfaced resources deployed today use
grid-following (GFL) control that assumes a stable ac voltage waveform (i.e., frequency and magnitude) at the converter terminal. While GFL converters can provide basic grid support functions (e.g., primary frequency control), their large-scale integration jeopardizes stability of power systems~\cite{MDH+18}. In contrast, GFM controls such as droop control \cite{CDA93}, virtual synchronous machine control \cite{DSF2015}, and (dispatchable) virtual oscillator control \cite{JD+2014,GCB+2019}, self-synchronize and impose a stable ac voltage waveform at the converter terminal. 

{Both GFL and GFM control need to account for semiconductor current limits. Broadly speaking, GFL converters are  current controlled and limiting the current reference does not severely degrade functionality~\cite{JY+2018}.} In contrast, GFM control under semiconductor current limits~\cite{DPD+2018,BCL+2024} remains a significant challenge. Existing overcurrent limiters are typically categorized into methods that limit current references and methods that manipulate voltage or power references to limit current (see \cite{BCL+2024} for an in-depth discussion).

An often overlooked aspect is the timescale and control layer at which a limiter acts. Four broad classes of limiters are shown in Fig.~\ref{fig:conv}. An outer GFM control typically provides a voltage reference (i.e., phase angle and magnitude) and inner controls may be present depending on the control architecture. {Due to its subcycle response and hardware implementation,} current limiting at the switch-level (see Fig.~\ref{fig:conv:switchlevel}) is {commonly seen as the most dependable approach to protecting semiconductor switches. However, it completely overrides the converter controls and hence is undesirable from a converter control and grid perspective}.

\begin{figure}[b!]
\subfloat[Converter with switch-level current limiter (e.g., pulse-blocking).\label{fig:conv:switchlevel}]{\makebox[1\columnwidth][c]{\includegraphics[width=0.85\columnwidth]{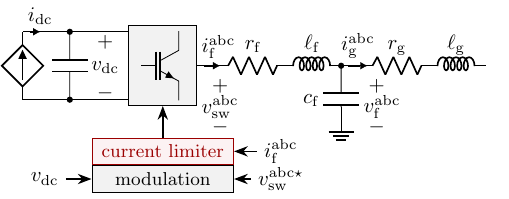}}}\hfill
\subfloat[Inner current limiting (e.g., current saturation) limits the current reference provided to an inner current loop with optional anti-windup (grey).\label{fig:conv:dualloop}]{\makebox[1\columnwidth][c]{\includegraphics[width=0.95\columnwidth]{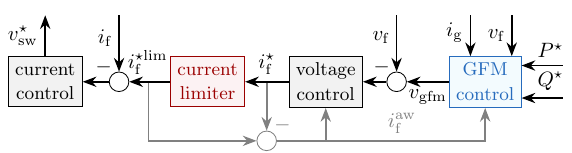}}}\hfill
\subfloat[Parallel current limiting (e.g., threshold virtual impedance) modifies the output $v_\textup{gfm}$ of the GFM control.\label{fig:conv:tvi}]{\makebox[1\columnwidth][c]{\includegraphics[width=0.7\columnwidth]{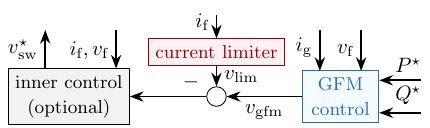}}}\hfill
\subfloat[Outer current limiting acts on the states of the GFM control through power (or frequency and voltage) references.\label{fig:conv:outer}]{\makebox[1\columnwidth][c]{\includegraphics[width=0.8\columnwidth]{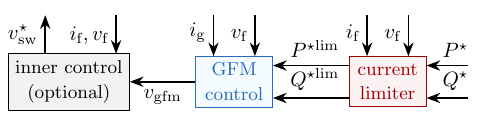}}}
\caption{Voltage source converter model and broad categorization of current limiting controls for GFM converters.\label{fig:conv}}
\end{figure}

Inner current limiting (see Fig.~\ref{fig:conv:dualloop}) limits current references tracked by underlying current controls and is prone to a loss of synchronization~\cite{XHZ+2016} {and short critical clearing times that may be challenging to meet by power system protection}. {Moreover, the presence of inner control loops can result in (i) challenging tuning problems under variable grid strength~\cite{QGC+18}, and (ii) significant bandwidth limitations in high-power converters with low switching frequency that impact both performance during nominal operation~\cite{GG2023} and fault ride through~\cite{DPD+2018}.}

To {retain some GFM features,} parallel current limiters (see Fig.~\ref{fig:conv:tvi}) such as threshold virtual resistance~\cite{FCG2010} and threshold virtual impedance~\cite{HL2011} modify the GFM voltage reference based on current measurements. However, the effectiveness may be reduced by bandwidth limitations of downstream current and voltage controls~\cite{DPD+2018}. Various variable and adaptive parallel limiters have been proposed to improve critical clearing time and stability margins~\cite{QWW+2023}. {Nonetheless, typical parallel limiters may be ineffective under large phase angle differences~\cite{ZBS+2023} and are, therefore, still subject to finite critical clearing time. While parallel hybrid limiters}  using voltage measurements can limit current when large phase angle differences occur (e.g., during fault clearing)~\cite{ZBS+2023}{, they may be ineffective under overcurrent induced by loads or large frequency deviations because they cannot significantly modify the converter frequency.}

Finally, outer limiters  (see Fig.~\ref{fig:conv:outer})  modify inputs to the GFM control to enhance transient stability {and typically do not exhibit transient stability challenges upon fault clearing (i.e., have large or infinite critical clearing time). However, they can result in significant current limit violations upon fault inception due to, e.g., bandwidth limitations of GFM and inner controls. In particular,} outer limiters, such as power angle limiting, are typically based on steady-state circuit models and {either (i) implicitly rely on switch-level current limiting to avoid exposing semiconductor switches to} transient overcurrent~\cite{AA2023,GL2023}{, or (ii)} combine outer limiting with inner current limiters~\cite{GD2019}. Outer limiters can further be categorized into methods that only use current measurements~\cite{GD2019,GL2023} and methods that only use voltage measurements~\cite{EUJ2022,AA2023} to modify the GFM voltage magnitude~\cite{CPO2020} or GFM voltage phasor~\cite{EUJ2022}. Notably, GFM converters relying on voltage-based limiting can resemble a controlled current source during current limiting which may lead to destabilizing effects in weak grids similar to GFL control~\cite[Sec.~III-B3]{BCL+2024}.

Recent works often combine parallel limiters, inner limiters, and anti-windup feedback~\cite{ALJ+2022,ARY+2023,HDH+2025} to retain as many GFM control features as possible under constraints. However, {these limiters may have many tuning parameters with complex interdependencies} and typically override upstream {GFM} control {transiently}, which raises stability concerns. {Moreover, combined methods typically focus on voltage phase angle forming~\cite{HDH+2025} and may exhibit challenges in scenarios that require significant modifications to phase angle or frequency (e.g., large frequency disturbances).}

Transient stability of GFM controls with inner current limiters can also be enhanced by mitigating oscillations between control modes using model predictive control with steady-state converter model and binary variables that model the activation of inner current limiters~\cite{AGP+2024}. On the other hand, predictive controls such as finite control set model predictive control focus on tracking current references (e.g., GFL control)~\cite{CKK+2008,QAP+2012} and their value for grid-connected converters is the subject of significant debate~\cite[Sec.-VIII]{KG2020}. Similarly, indirect model predictive control for power electronics is typically applied to current reference tracking~\cite{RKC+2022}.

{Finally, the aforementioned methods exclusively focus on current limits and do not provide a pathway to systematic handling of multiple constraints (e.g., current, modulation, power). In contrast, constraint-aware GFM control  
\begin{enumerate}
\item does not require inner control loops and avoids their sensitivity to grid strength and bandwidth limitations,
\item has few control parameters with clear interpretation,
\item exhibits infinite critical clearing time and is effective under phase jumps, overload, as well as large frequency perturbations, and
\item provides a pathway to systematically accounting for a wide range of converter constraints.
\end{enumerate}
}

{While asymmetric faults are significantly more common in transmission systems~\cite{EAC+1992} and have received significant attention in the recent literature~\cite{REL2021,BVKS+2022,ARY+2023,BG2023,HDH+2025}, this work revisits the symmetric case. In particular, the goal of this work is} to develop a systematic foundation for {GFM control under constraints that resolves conceptual and practical challenges associated with existing current limiters}. 
Notably, existing works (i) do not start from a specification of how a GFM converter should respond when encountering a limit, and (ii) augment GFM controls with specific ad-hoc implementations of current limiters that act in parallel or downstream of unconstrained GFM controls. 

Instead, this work formalizes the specification that a GFM converter should minimize the deviation of the converter from unconstrained GFM control dynamics (e.g., \cite{CDA93,DSF2015,JD+2014,GCB+2019}) when encountering a converter constraint (e.g., modulation or current limit) and explicitly embed a constraint model in the GFM control. To this end, the first contribution of this work is to develop and analyze a lightweight prediction model to map current constraints to converter voltage constraints. This construction of the constraint set naturally incorporates both filter current and voltage measurements into the current limiter.

The second contribution is an efficient numerical algorithm for projecting voltages onto the constraint set, i.e., given a reference voltage obtained from unconstrained GFM control the algorithm computes the closest voltage in the constraint set. These two preliminary results are the basis for the main contribution of this work, a constraint-aware GFM droop control that explicitly incorporates constraint handling. The performance, robustness, and dynamic interactions of converters are studied using simulations. Finally, hardware experiments at various short-circuit ratios are used to validate and illustrate the results. Notably, the proposed control is able to limit current near instantaneously and maintain operation within current limits and synchronization under severe events (e.g., $180$-degree phase jumps). Moreover, in contrast to control architectures that rely on inner current and voltage controls, the tuning of the proposed control is largely independent of the grid strength.

This manuscript is organized as follows. Section~\ref{sec:objective}  introduces converter constraints and control objectives. The set of feasible voltages is constructed in Sec.~\ref{sec:feas} and an efficient algorithm for projecting onto the set of feasible voltages is developed in Sec.~\ref{sec:proj}. The constraint-aware GFM control is presented in Sec.~\ref{sec:gfm}. The control is evaluated in simulation in Sec.~\ref{sec:scib} and Sec.~\ref{sec:twobus}. Finally, hardware experiments are presented in Sec.~\ref{sec:experiment} and conclusions are provided in Sec.~\ref{sec:concl}

\subsection*{Notation} We use $\mathbb{R}$ and $\mathbb N$ to denote the set of real and natural numbers and define, e.g., $\mathbb{R}_{\geq 0}\coloneqq \{x \in \mathbb R \vert x \geq 0\}$. For column vectors $x\in\mathbb{R}^n$ and $y\in\mathbb{R}^m$ we define $(x,y) \coloneqq [x^\mathsf{T}, y^\mathsf{T}]^\mathsf{T} \in \mathbb{R}^{n+m}$ and $I_n$, $\mathbbl{0}_{n\times m}$, and $\mathbbl{0}_{n}$, denote the $n$-dimensional identity matrix, $n \times m$ zero matrix, and column vectors of zeros of length $n$. Given a vector $x\in\mathbb{R}^n$ and positive definite matrix $W \in \mathbb{R}^{n \times n}$, the weighted Euclidean norm is denoted by $\norm{x}_W \coloneqq \sqrt{x^\mathsf{T} W x}$. The angle of a vector $x \in \mathbb{R}^2$ is denoted by $\angle \, x$. The two-dimensional rotation matrix $\mc R_\text{2D}$ with angle $\theta_\text{dq} \in [-\pi,\pi)$ and the amplitude-invariant Clarke transformation $\mc T_{\alpha\beta}$ are given by
\begin{align*}
\mc R_\text{2D}(\theta)\coloneqq  \begin{bmatrix} \cos(\theta) & -\sin(\theta) \\ \sin(\theta) & \cos(\theta)\end{bmatrix}\!, \; T_{\alpha\beta} \coloneqq \frac{2}{3} {\begin{bmatrix} 1  &  -\tfrac{1}{2}  &  -\tfrac{1}{2} \\  0 &  \tfrac{\sqrt 3}{2}  & -\frac{\sqrt 3}{2} \end{bmatrix}}\!.
\end{align*}
Finally, we define the $90^\circ$ rotation matrix $j \coloneqq \mc R_\text{2D}(\frac{\pi}{2})$ and the dq-transformation $T_\text{dq}(\theta_\text{dq} ) \coloneqq {\mc R_\text{2D}(-\theta_\text{dq})   T_{\alpha\beta}}$.


\section{Converter Model and Control Objectives}\label{sec:objective}

\subsection{Voltage Source Converter Model}
To illustrate our control, we will use a two-level voltage source converter (VSC) with LCL filter shown in Fig.~\ref{fig:conv:switchlevel}  with modulated voltage $v_\text{sw} \in \mathbb{R}^3$, filter current $i^\text{abc}_\text{f} \in \mathbb{R}^3$, filter voltage $v^\text{abc}_\text{f} \in \mathbb{R}^3$, and grid current  $i^\text{abc}_\text{g} \in \mathbb{R}^3$. {Throughout this manuscript the dc voltage is assumed to be constant. This setup, in abstraction, models a converter connected to energy storage and is widely considered for studies of GFM converters.} We emphasize that the control developed in this work does not require that the filter capacitor $C_\text{f}$ or grid-side inductor $L_\text{g}$ are present. In other words, the control can be applied to VSCs that only feature phase reactors (e.g., modular multi-level converters). In the remainder, we focus on balanced three-phase systems. Extensions to unbalanced three-phase systems are seen as an interesting topic for future work. Given an electrical signal $x^\text{abc} \in \mathbb{R}^3$ we define its representation in $\alpha\beta$-frame and dq-frame with angle $\theta_\text{dq} \in [-\pi,\pi)$ as $x^{\alpha\beta}=T_{\alpha\beta} x^\text{abc}$ and $x^{\text{dq}}=T_\text{dq}(\theta_\text{dq}) x^\text{abc}$. Finally, we define the instantaneous active and reactive power
\begin{align}
P_\text{f} \coloneqq  \tfrac{3}{2}  (v^{\alpha\beta}_\text{f})^\mathsf{T}  i^{\alpha\beta}_\text{f}, \quad Q_\text{f} \coloneqq  \tfrac{3}{2} (v^{\alpha\beta}_\text{f})^\mathsf{T}  j i^{\alpha\beta}_\text{f}.
\end{align}

\subsection{Voltage Source Converter Constraints}
This work focuses on ac current limits and modulation constraints. In particular, we define the current limit $i_{\max} \in \mathbb{R}_{>0}$ and modulation limit $V_{\max} \coloneqq \frac{1}{2} v_\text{dc} \in \mathbb{R}_{>0}$ of a three-phase two-level VSC. In particular, we aim to enforce the long-term current limit $\|i^{\alpha\beta}_\text{f}(t)\| \leq i_{\max}$ and modulation limit $\|v^{\alpha\beta\star}_\text{sw}\| \leq V_{\max}$ at the GFM control level to prevent pulse blocking and clipping of voltage references by the modulator. Typically, brief transient violations (i.e., on the order of tens of milliseconds) of the long-term current limit are permissible up to the short-term current limit  $i^\text{shrt}_{\max} > i_{\max}$, i.e., $\|i^{\alpha\beta}_\text{f}(t)\| \leq i^\text{shrt}_{\max}$.

\subsection{{Control Objectives and Architecture}}
Consider frequency and voltage magnitude references $\omega_\text{dr}\in\mathbb{R}$ and $V_\text{dr}\in\mathbb{R}_{\geq 0}$ obtained from the droop equations
\begin{subequations}
\begin{align}
\omega_\text{dr} &= \omega_0 + m_P (P^\star - P_\text{lp}),\\
V_\text{dr} &= V^\star + m_Q (Q^\star - Q_\text{lp}),
\end{align}
\end{subequations}
with active and reactive power droop coefficients $m_p \in \mathbb{R}_{>0}$ and $m_q \in \mathbb{R}_{>0}$, nominal frequency $\omega_0 \in \mathbb{R}_{>0}$, setpoints $V^\star \in \mathbb{R}_{\geq 0}$, $P^\star \in \mathbb{R}$, $Q^\star \in \mathbb{R}$ for the voltage magnitude, active power, and reactive power, and low-pass filtered active and reactive power measurements $P_\text{lp} \in \mathbb{R}$ and $Q_\text{lp} \in \mathbb{R}$ (e.g., to eliminate unwanted harmonics~\cite{HNL+2015}). 

Next, consider the voltage phase angle $\theta \in \mathbb{R}$ and magnitude $V \in \mathbb{R}_{\geq 0}$  determined by the GFM droop controller
\begin{subequations}\label{eq:droop:cont}
\begin{align}
\ddt \theta &= \omega_0 + m_P (P^\star - P_\text{lp}),\label{eq:droop:cont:freq}\\
\tau_\text{v} \ddt V  &= -V + V^\star + m_Q (Q^\star - Q_\text{lp}),\label{eq:droop:cont:volt}
\end{align}
\end{subequations}
with voltage control time constant $\tau_\text{v} \in \mathbb{R}_{>0}$. The control objective considered in this work is to minimize the deviation from \eqref{eq:droop:cont} when the VSC encounters a constraint. 

To this end, consider a discrete-time controller implementation with sampling time $\tau_\text{ctr} \in \mathbb{R}_{>0}$ evaluated at times $t_k \coloneqq k \tau_\text{ctr}$ for all $k\in \mathbb{N}_0$. At every sampling time $t_k$, the discrete-time implementation
\begin{subequations}\label{eq:droop:disc}
\begin{align}
\hat{\theta}(t_k) &= \theta(t_{k-1}) + \tau_{\text{ctr}} \omega_\text{dr}(t_k),\label{eq:droop:disc:theta}\\
\hat{V}(t_k) &= A_{\tau_\text{v}} V(t_{k-1})  + B_{\tau_\text{v}} V_\text{dr}(t_k).\label{eq:droop:disc:volt}
\end{align}
\end{subequations}
of \eqref{eq:droop:cont} with $A_{\tau} \coloneqq e^{-\frac{\tau_\text{ctr}}{\tau}}$ and $B_{\tau} \coloneqq 1-A_{\tau}$ is used to compute a candidate GFM voltage phase angle $\hat{\theta}(t_k)$ and magnitude $\hat{V}(t_k)$. Next, consider the set $\mc C(t_k)$ that encodes the converter constraints (see Sec.~\ref{sec:feas}). The projection 
\begin{align*}
(\theta(t_k),V(t_k)) = &\argmin_{(\tilde{\theta},\tilde{V}) \in \mc C(t_k)} \; w_\theta (\tilde{\theta} - \hat{\theta}(t_k))^2+ (\tilde{V} \!-\! \hat{V}(t_k))^2
\end{align*}
computes the closest voltage phase angle and magnitude to the candidate GFM voltage phase angle and magnitude that is within the converter constraints. The trade-off between voltage phase angle and magnitude modifications is encoded by the weight $w_\theta \in \mathbb{R}_{>0}$. Notably, if the unconstrained GFM controller \eqref{eq:droop:disc} does not violate the converter limits, then $(\theta(t_k),V(t_k))=(\hat{\theta}(t_k),\hat{V}(t_k))$.

\begin{figure}[b!]
\centering
\includegraphics[width=1\columnwidth]{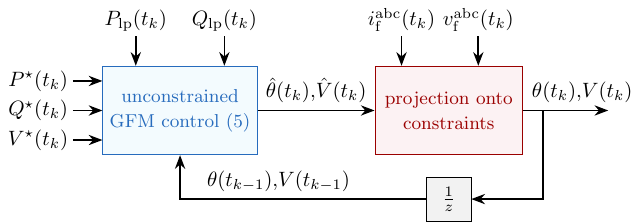}
\caption{Constraint-aware GFM control.\label{fig:proj.arch}}
\end{figure}

The overall control architecture is shown in Fig.~\ref{fig:proj.arch} and explicitly integrates a model of the constraints into the GFM controller instead of relying on inner, parallel, or outer limiters. Notably, the proposed control does not require inner control loops and avoids potential performance degradation due to their bandwidth limits~\cite{GG2023,DPD+2018}. Moreover, constructing the constraint set $\mc C(t_k)$ does not require knowledge of the grid reactance-to-resistance ratio or grid voltage (see Sec.~\ref{sec:feas}). Finally, the constraint-aware GFM control in Fig.~\ref{fig:proj.arch} can systematically account for converter limits beyond current limits (see Sec.~\ref{sec:modlimits} and Sec.~\ref{sec:powerlimit}) and has a straightforward interpretation relative to grid-level specifications (see Sec.~\ref{sec:gfm:disc}).

\subsection{Review of optimal virtual impedance current limiting} 
The objective of minimizing the voltage deviation induced by virtual impedance current limiters (e.g., Fig.~\ref{fig:conv:tvi}) has been considered in the literature. Before reviewing these results, we briefly review the basic threshold virtual impedance current limiter~\cite{HL2011,QWW+2023}. Threshold virtual impedance emulates a converter output impedance with reactance-to-resistance ratio $\rho_\frac{X}{R} \in \mathbb{R}_{\geq0}$ that increases with gain $k_\text{VI}\in\mathbb{R}_{> 0}$ when $\| i^\text{dq}_\text{f}\|$ is larger than the threshold $i_\text{thr} < i_{\max}$. In other words, in dq frame with angle $\theta$, the unconstrained GFM controller \eqref{eq:droop:cont} voltage reference $v^\text{dq}_\text{gfm}=(V,0) \in \mathbb{R}^2$ is modified by
\begin{align}\label{eq:TVI}
v^{\text{dq}}_\text{sw} = v^\text{dq}_\text{gfm} - k_\text{VI} \max(0,\norm{i^\text{dq}}-i_\text{thr}) (I_2 +j \rho_{\frac{X}{R}}) i^\text{dq}_\text{f}.
\end{align}
Notably, for a converter connected to a constant voltage source behind a grid impedance (i.e., Th\'evenin equivalent), the voltage deviation $\| v^\text{dq}_\text{sw}-v^\text{dq}_\text{gfm}\|$ for a given steady-state current magnitude is minimized by setting $\rho_{\frac{X}{R}}$ equal to the grid reactance-to-resistance ratio~\cite[Sec.~V]{GD2019},\cite[Sec.~III]{WWZ2024}. However, the (equivalent) grid reactance-to-resistance ratio is generally unknown. Moreover, we note that, for any fixed current magnitude, the voltage modification in \eqref{eq:TVI} is proportional to the current. As a consequence, challenges arise during scenarios that require rapid phase angle changes~\cite{ZBS+2023} or significant changes to frequency  (see Sec.~\ref{sec:comp}).

\section{The Set of Feasible Converter Voltages}\label{sec:feas}
For the purpose of control design, we map all converter limits to constraints on the voltage reference $v_\text{sw}^{\alpha\beta\star}\in \mathbb{R}^2$ provided to the modulator in Fig.~\ref{fig:conv:switchlevel}. To this end, consider the GFM voltage reference $v^\text{dq}_\text{gfm} \in \mathbb{R}^2$ in a dq frame with angle $\theta_\text{dq} \in [-\pi,\pi)$, frequency $\ddt \theta_\text{dq} = \omega_\text{dq}$, and a voltage $v^{\alpha\beta}_\text{ad} \in \mathbb{R}^2$ in stationary frame provided by auxiliary (e.g., active damping) controls. Then, we parameterize $v_\text{sw}^{\alpha\beta\star}$ as
\begin{align}\label{eq:ctrinput}
v_\text{sw}^{\alpha\beta\star}(t) \coloneqq  \mc R_\text{2D}(\theta_\text{dq}(t))  v^\text{dq}_\text{gfm}(t)  - v_\text{ad}^{\alpha\beta}(t).
\end{align}

\subsection{Modulation Limits}\label{sec:modlimits}
By substituting \eqref{eq:ctrinput} into $\| v_\text{sw}^{\alpha\beta\star}(t_k)\| \leq V_{\max}$ the modulation limit can be equivalently expressed as disc constraint
\begin{align}\label{eq:cons.disc.volt}
\mc C^{\alpha\beta}_{\text{mod}} (t) \!\coloneqq\! \left\{ v^{\alpha\beta}_\text{gfm} \in \R^2 \left\vert \; \|v^{\alpha\beta}_\text{gfm}  -c^{\alpha\beta}_{\text{mod}}(t)\| \leq r_\text{mod}\!\right. \right\}
\end{align}
with $c_\text{mod}(t) \coloneqq v^{\alpha\beta}_\text{ad}(t)$ and  $r_\text{mod} \coloneqq V_{\max}$. In other words, modulation limits directly apply to the GFM voltage reference.

\subsection{Current Limits}
In contrast, the filter current $i^\text{abc}_\text{f}$ is a state of the system and cannot be instantaneously limited by the control. Instead, we map the current limit to the voltage reference provided by the GFM control.
To this end, let $r_\text{f}  \in \mathbb{R}_{>0}$,  $\ell_\text{f}  \in \mathbb{R}_{>0}$, and $Z_\text{f}  \coloneqq r_\text{f}  I_2 + j \omega_\text{dq} \ell_\text{f}$ denote the filter resistance, filter inductance, and filter impedance matrix. The dynamics of the filter current with $v_\text{sw}^{\text{abc}\star}(t)$ given by \eqref{eq:ctrinput} are given by $v_\text{ad}^{\text{dq}} = T_{\text{dq}}(\theta_\text{dq}) v_\text{ad}^{\alpha\beta}$ and
\begin{align}\label{eq:if.cont}
L_\text{f}  \ddt i^\text{dq}_\text{f}(t)  = -Z_\text{f}  i^\text{dq}_\text{f}(t)  + v^\text{dq}_\text{gfm}(t) - v_\text{ad}^{\text{dq}}(t) - v^\text{dq}_\text{f}(t).
\end{align}
Given a time $\tau \in \mathbb{R}_{>0}$, and assuming that $v^\text{dq}_\text{ad}$ and $v^\text{dq}_\text{f}$ are constant from $t$ to $t+\tau$, a forward prediction of the current at time $t+\tau$ computed at time $t$ is given by
\begin{align}\label{eq:if.disc.ctr}
i^{\text{dq}}_{\text{f},\tau} (t) \!\coloneqq\! A_\tau i^\text{dq}_\text{f}(t)  + B_\tau (v^\text{dq}_\text{gfm}(t) - v^\text{dq}_\text{ad}(t) - v^\text{dq}_\text{f}(t))
\end{align}
with $A_\tau \coloneqq e^{-\frac{r_\text{f}}{l_\text{f}}\tau} \mc R_\text{2D}(-\omega_\text{dq} \tau)$ and $B_\tau \coloneqq Z^{-1}_\text{f} (I_2 - A_\tau)$. 
%
%
By substituting \eqref{eq:if.disc.ctr}, {the constraint $\| i^{\text{dq}}_{\text{f},\tau}(t) \| \leq i_{\max}$ on the predicted current can be equivalently expressed in $\alpha\beta$ frame as time-varying disc constraint}
\begin{align}\label{eq:cons.disc}
\mc C^{\alpha\beta}_{\text{cur},\tau}(t) \!\coloneqq\! \left\{ v^{\alpha\beta}_\text{gfm} \in \R^2 \left\vert \; \|v^{\alpha\beta}_\text{gfm}  -c^{\alpha\beta}_{\text{cur},\tau}(t)\| \leq r_{\text{cur},\tau}\!\right. \right\}
\end{align}
{on the GFM voltage reference with center $c_{\text{cur},\tau}(t) \coloneqq v^{\alpha\beta}_\text{f}(t)+v^{\alpha\beta}_\text{ad}(t)-M_\tau i^{\alpha\beta}_\text{f}(t)$, $M_\tau \coloneqq (A^{-1}_\tau  - I_2)^{-1} Z_\text{f}$, and radius}
\begin{align}\label{eq:kappa}
r_{\text{cur},\tau}  \coloneqq i_{\max} \sqrt{\frac{\ell_\text{f}^2 \omega_\text{dq}^2 + r_\text{f}^2}{1+e^{-2\frac{r_\text{f}}{\ell_\text{f}} \tau} - 2 e^{-\frac{r_\text{f}}{\ell_\text{f}} \tau} \cos(\omega_\text{dq} \tau)}}.
\end{align}
{Notably, by construction of $A_\tau$ and $B_\tau$, the constraint  \eqref{eq:cons.disc} assumes a constant frequency $\omega_\text{dq}$ over the prediction horizon $\tau$. This assumption is typically justified for transmission systems with slowly changing frequency. Please see Sec.~\ref{subsec:lowfr} for simulations of a sudden frequency drop.}

\begin{remark}[\textbf{Robustness and inductor nonlinearities}]\label{rem:nonlin}
Because \eqref{eq:if.cont} only needs to be accurate near the maximum converter current, we use the inductance $\ell_\textup{f}$ at the rated inductor current in the predictor \eqref{eq:if.cont}. This value is typically known to within a few percent. To enable prolonged operation at the current limit, we assume the rated current of the inductor to be equal to the long-term current limit $i_{\max}$ of the converter. In addition, $r_\textup{f}$ is interpreted as the equivalent series resistance (ESR) of the converter and filter to model core losses~\cite{MSB+2018} and switching losses~\cite{PCM1980}. In particular, for our hardware experiments we use $r_\textup{f} = P_\textup{loss}/i^2_{\max}$, where $P_\textup{loss}=P_\textup{dc}-P_\textup{f}$ is the overall converter loss obtained from an experiment at rated voltage, maximum current, and unit power factor.
\end{remark} 

\subsection{{Constraints Beyond Modulation and Current Limits}}\label{sec:powerlimit}
{While this work focuses on modulation and current limits, a wide range of limits can be modeled as constraints on the VSC voltage. For example, active and reactive power limits $P_{\min},P_{\max} \in \mathbb{R}$ and $Q_{\min},Q_{\max} \in \mathbb{R}$ can be modeled as
\begin{subequations}\label{eq:PQcons}
	\begin{align}
	\! \! \mc C^{\alpha\beta}_{P_\text{sw}} (t) \!&\coloneqq\! \left\{ v^{\alpha\beta}_\text{gfm} \!\in \!\R^2 \left\vert P_{\min}\!\leq \tfrac{3}{2} (i^{\alpha\beta}_\text{f})^\mathsf{T} v^{\alpha\beta}_\text{gfm}  \leq \! P_{\max} \! \right. \right\}\!,\\
	\! \!  \mc C^{\alpha\beta}_{Q_\text{sw}} (t) \!&\coloneqq\! \left\{ v^{\alpha\beta}_\text{gfm} \!\in \!\R^2 \left\vert Q_{\min}\! \leq \tfrac{3}{2}(j i^{\alpha\beta}_\text{f})^\mathsf{T} v^{\alpha\beta}_\text{gfm} \leq \! Q_{\max} \! \right. \right\}\!.
	\end{align}
\end{subequations} 
Moreover, dc voltage limits can be mapped to the GFM voltage using a prediction model of the dc capacitor dynamics. Exploring constraints beyond modulation and current limits is seen as an interesting topic for future work.}

\subsection{The Set of Feasible GFM Voltages}\label{subsec:feasibleset}
While the GFM control should satisfy the current and voltage limits at every time $t_k$, a significant question is for what prediction horizons $\tau \in \mathbb{R}_{>0}$ to enforce the constraint \eqref{eq:cons.disc}. A common approximation is to only enforce the constraint for the predicted current at the next sampling time, i.e., $v^{\alpha\beta}_\text{gfm}(t_k) \in \mc C^{\alpha\beta}_{\text{cur},\tau_\text{ctr}}(t_k)$, i.e., to enforce operating within the current limits at all controller sampling times.

However, GFM standards and grid codes currently under discussion are expected to contain the additional requirement to reach a steady-state voltage within one cycle after an event (e.g., short-circuit fault). To meet this requirement, a natural choice may be to require $v^{\alpha\beta}_\text{gfm}(t_k) \in \mc C^{\alpha\beta}_{\text{cur},\tau}(t_k)$ for all $\tau \in [0,\frac{1}{f_0}]$, where $f_0 \in \mathbb{R}_{>0}$ is the nominal frequency of the ac system. However, this constraint is infinite dimensional and not tractable in the context of real-time control of grid-connected converters. Instead, in the remainder of this paper, we impose the constraint $v^{\alpha\beta}_\text{gfm}(t_k) \in \mc C^{\alpha\beta}_{\text{cur},\tau_\text{cyc}}(t_k) \cap \mc C^{\alpha\beta}_{\text{cur},\tau_\text{cyc}}(t)$, with $\tau_\text{cyc} \approx \frac{1}{f_0}$. Broadly speaking, this constraint requires that a GFM voltage with phase angle $\angle\, v^{\alpha\beta}_\text{gfm}(t_k)$, magnitude $\| v^{\alpha\beta}_\text{gfm}(t_k)\|$, and constant frequency $\omega_\text{dq}$ is predicted to meet the current limit one time step and one cycle ahead, if all other quantities (i.e., $v^{\alpha\beta}_\text{f}$ and $v^{\alpha\beta}_\text{ad}$) are in a steady-state with frequency $\omega_\text{dq}$.

During transients the frequency of $v^{\alpha\beta}_\text{f}$ and $v^{\alpha\beta}_\text{ad}$  cannot be expected to match the frequency $\omega_\text{dq}$. However, obtaining more accurate predictions of future currents and GFM voltage frequencies would require full knowledge of the entire grid model and state that are not available to the converter control.

Overall, we consider the constraint 
\begin{align}\label{eq:consset}
\mc C^{\alpha\beta} \coloneqq \mathcal{C}^{\alpha\beta}_\text{mod} \cap \mathcal{C}^{\alpha\beta}_{\text{cur},\tau_\text{ctr}} \cap \mathcal{C}^{\alpha\beta}_{\text{mod},\tau_\text{cyc}}.
\end{align}
The constraint set for a converter riding through a short-circuit fault (captured at $t=0.45$~s of the simulation in Sec.~\ref{sec:scib}) is shown in Fig.~\ref{fig:feas}.
\begin{figure}[b!]
\centering
\includegraphics[width=0.9\columnwidth]{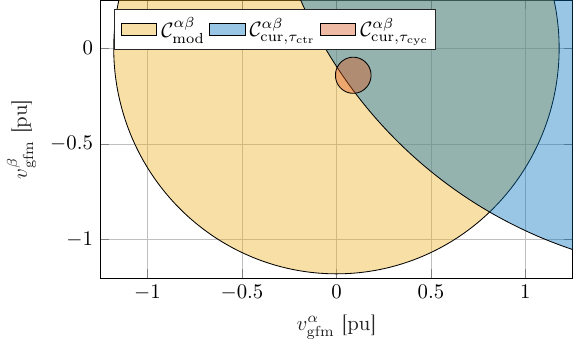}
\caption{The set of feasible voltages during a symmetric short-circuit fault.\label{fig:feas}}
\end{figure} 
It can be seen that, for a significant share of the voltages $v^{\alpha\beta}_\text{gfm} \in \mc C^{\alpha\beta}_{\text{cur},\tau_\text{ctr}}$,  that are predicted to satisfy the current limit at the next controller sampling time, no steady-state voltage with frequency $\omega_\text{dq}$ exists within the current limits (i.e., $v^{\alpha\beta}_\text{gfm} \notin \mc C^{\alpha\beta}_{\text{cur},\tau_\text{cyc}}$).

The following proposition provides a condition under which $\mc C^{\alpha\beta}$ is non-empty, i.e., a feasible voltage $v^{\alpha\beta}_\text{f}$ exists.

\begin{proposition}[\textbf{Non-empty feasible set}]\label{prop:nonempty}
Assume that $\| i^{\alpha\beta}_\textup{f}(t_k) \| \leq i_{\max}$ and $v^{\alpha\beta}_\textup{f}(t_k) + v^{\alpha\beta}_\textup{ad}(t_k) \in \mc C^{\alpha\beta}_\textup{mod}(t_k)$. Then, $\mc C^{\alpha\beta}(t_k)$ is non-empty.
\end{proposition}
\begin{proof}
Leting $v^{\alpha\beta}_\text{gfm}(t_k) = v^{\alpha\beta}_\textup{f}(t_k) + v^{\alpha\beta}_\textup{ad}(t_k)$, the prediction \eqref{eq:if.disc.ctr} simplifies to $i^{\text{dq}}_{\text{f},\tau}(t_k) = A_\tau i^\text{dq}_\text{f}(t_k)$ for all $\tau \in \mathbb{R}_{>0}$. Moreover, by definition of $A_\tau$, it holds that $\| i^{\text{dq}}_{\text{f},\tau} (t_k) \| \leq e^{-\frac{r_\text{f}}{l_\text{f}} \tau} \| i^\text{dq}_\text{f} (t_k)\|$ holds for all $\tau \in \mathbb{R}_{>0}$. Under the conditions of the proposition, $\| i^\text{dq}_\text{f} (t_k)\| = \| i^{\alpha\beta}_\text{f} (t_k)\| \leq i_{\max}$. This implies $v^{\alpha\beta}_\textup{f}(t_k) + v^{\alpha\beta}_\textup{ad}(t_k) \in \mathcal{C}^{\alpha\beta}_{\text{mod},\tau}(t_k)$ for any $\tau \in \mathbb{R}_{>0}$ and the proposition immediately follows from $v^{\alpha\beta}_\textup{f}(t_k) + v^{\alpha\beta}_\textup{ad}(t_k) \in \mc C^{\alpha\beta}_\textup{mod}(t_k)$.
\end{proof}
{Thus,} as long as the {filter current is within its limit and} the sum of the filter voltage and active damping control is within the modulation limit $\mc C^{\alpha\beta}_{\text{mod}}(t_k)$, {there always exists} a voltage $v^{\alpha\beta}_\textup{f}(t_k)$ within the modulation limits such that the predicted current is within {its} limit. Notably, this result also implies feasibility at the next controller sampling time unless $v^{\alpha\beta}_\textup{f}(t_{k+1}) + v^{\alpha\beta}_\textup{ad}(t_{k+1})$ exceeds the modulation limit. 

Proposition~\ref{prop:nonempty} also highlights that continuous operation within the constraints is generally not possible during grid or load induced overvoltage that exceeds the converter modulation limits as the converter can no longer fully control its output current. More generally, if the set $\mc C^{\alpha\beta}(t_k)$ is empty, no voltage reference $v^{\alpha\beta}_{\text{sw}}$ with constant frequency over $\tau_\text{cyc}$ and magnitude within the modulation limits exist that does not result in violation of current limits. This is true  independently of the specific control strategy. In this case, it appears reasonable to employ pulse blocking for brief periods of time (e.g., when brief transient overvoltages exceed the converter modulation limit) until $\mc C^{\alpha\beta}(t_k)$ is non-empty. If $\mc C^{\alpha\beta}(t_k)$ is empty for an extended period of time to be defined, prolonged overvoltage or overcurrent usually require disconnecting the converter. In other words, the proposed framework can be used to systematically detect scenarios in which to allow grid-connected converters to momentarily cease operation or trip.

\begin{remark}[\textbf{Low switching and sampling frequencies}]\label{rem:switching}
The low switching frequency of converters in high-power applications significantly reduces control bandwidth~\cite{GG2023} and degrades the performance of common current-limiting methods for GFM controls~\cite{DPD+2018}. To improve control performance, predictive delay compensators~\cite{LWL+2018} may be used. Using the predictor \eqref{eq:if.cont} embedded into the constraint set \eqref{eq:consset} to compensate for increased latency is seen as an interesting topic for future work. While the accuracy of \eqref{eq:if.cont} decreases with increasing $\tau$, lowering the switching frequency and sampling rate is not expected to significantly degrade the accuracy of \eqref{eq:if.cont}. In particular, $\frac{1}{f_0} = \tau_\text{\upshape{cyc}} \gg \tau_\text{\upshape{ctr}} = \frac{1}{f_\text{\upshape{sw}}}$ when $f_\text{\upshape{sw}} \gg f_0$, i.e., the length of the prediction horizon is not dictated by the switching frequency but by grid-level requirements (see Sec.~\ref{sec:gfm:disc}).
\end{remark} 

\section{Projection onto the Feasible Set}\label{sec:proj}
Next, given an infeasible candidate voltage provided by an unconstrained GFM control, we consider the problem of finding the feasible voltage with minimal distance to the infeasible candidate voltage. We consider two different coordinate frames (polar and Cartesian) that result in different notions of distance between two voltages. While the projection in polar coordinates is more closely aligned with control objectives in power systems (see Sec.~\ref{sec:gfm:disc}), the projection in Cartesian coordinates is tractable for real-time control. For brevity of the notation the time index $t_k$ is omitted in this section.

\subsection{Polar Coordinates}
We first consider a candidate voltage in polar coordinates with phase angle $\hat{\theta} \in [-\pi,\pi)$, magnitude $\hat{V} \in \mathbb{R}_{\geq 0}$, and $v^\text{dq}_{\text{gfm}}=(\hat{V},0) \in \mathbb{R}^2$. The trade-off between modifying the voltage phase angle and the voltage magnitude is parameterized by $w_\theta \in \mathbb{R}_{\geq 0}$. The oblique projection of the pair $(\hat{\theta},\hat{V})$ onto the feasible set $\mc C^{\alpha\beta}$ in polar coordinates is given by
\begin{subequations}\label{eq:proj:polar} 
\begin{align}
\Pi^{\text{pol}}_{\mc C^{\alpha\beta}} (\hat{\theta},\hat{V}) \coloneqq &\argmin_{\theta,V} \; w_\theta (\theta - \hat{\theta})^2+ (V - \hat{V})^2 \label{eq:proj:polar:cost}\\
&\; \text{s.t. }  \quad \mc R_{\text{2D}}(\theta) \begin{bmatrix} V \\ 0 \end{bmatrix} \in \mc C^{\alpha\beta} \label{eq:proj:polar:cons}
\end{align}
\end{subequations}
In other words, increasing $w_\theta \in \mathbb{R}_{> 0}$ results in giving priority to reducing phase angle modifications. Mapping the constraint set shown in Fig.~\ref{fig:feas} to polar coordinates results in the non-convex constraints shown in Fig.~\ref{fig:Cpolar}. The projection $\Pi^{\text{pol}}_{\mc C^{\alpha\beta}} (\hat{\theta},\hat{V})$ of the (unconstrained) reference voltage $(\hat{\theta},\hat{V})$ onto the constraint set is illustrated in Fig.~\ref{fig:Cpolar}. Notably, the projection \eqref{eq:proj:polar} is a nonlinear non-convex program that may be very challenging to solve in real time. 
\begin{figure}[b]
\centering
\includegraphics[width=0.95\columnwidth]{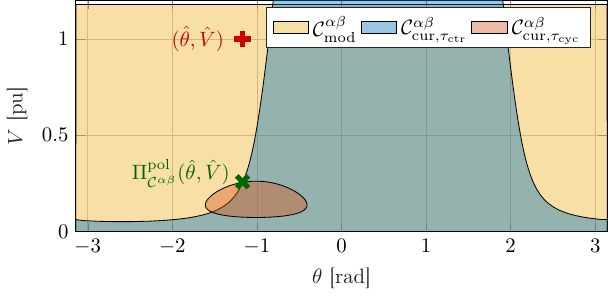}
\caption{Set of feasible voltages during a short-circuit fault: reference voltage provided by GFM droop control (\textcolor{red!80!black}{$\Plus$}) and its projection onto the set of feasible voltages (\textcolor{green!40!black}{$\Xmark$}). \label{fig:Cpolar}}
\end{figure}

\subsection{Cartesian Coordinates}
To approximate \eqref{eq:proj:polar} in Cartesian coordinates, we first note that, for $v^{\alpha\beta}_\text{gfm} = \mc R(\hat{\theta}) v^\text{dq}_\text{gfm}$ and any $v^\text{dq}_\text{gfm} \in \mathbb{R}^2$, it holds that $\angle\, v^{\alpha\beta}_\text{gfm} = \hat{\theta} + \angle\, v^\text{dq}_\text{gfm}$. Using $\theta = \angle\, v^{\alpha\beta}_\text{gfm}$, it holds that
\begin{align}\label{eq:cost_cart}
\!\!w_\theta (\theta - \hat{\theta})^2+ (V - \hat{V})^2 \!=\! w_\theta (\angle v^\text{dq}_\text{gfm})^2 + (\| v^\text{dq}_\text{gfm}  \| - \hat{V} )^2
\end{align}
Moreover, the second-order Taylor series expansion of \eqref{eq:cost_cart} at $\hat{v}^\text{dq}_\text{gfm} = (\hat{V},0)$ is given by $\| v^\text{dq}_\text{gfm} - \hat{v}^\text{dq}_\text{gfm}\|_{W}^2$ with $W \coloneqq \diag\{1,w_\theta / \hat{V}^2\} \in \mathbb{R}^{2 \times 2}$. Next, we transform the constraints to a dq frame with angle $\hat{\theta} \in [-\pi,\pi)$, i.e., $c^\text{dq}_\text{mod}(\hat{\theta}) \coloneqq \mc R_\text{2D} (-\hat{\theta}) c^{\alpha\beta}_\text{mod}$ and $c^\text{dq}_{\text{cur},\tau}(\hat{\theta}) \coloneqq \mc R_\text{2D} (-\hat{\theta}) c^{\alpha\beta}_{\text{cur},\tau}$ as well as $\mc C^\text{dq}_\text{mod}(\hat{\theta}) \coloneqq  \{ v^\text{dq}_\text{gfm} \in \R^2 \vert \; \|v^\text{dq}_\text{gfm}  -c^\text{dq}_\text{mod}(\hat{\theta})\| \leq r_\text{mod}  \}$ and $\mc C^\text{dq}_{\text{cur},\tau}(\hat{\theta}) \coloneqq  \{ v^\text{dq}_\text{gfm} \in \R^2 \vert \; \|v^\text{dq}_\text{gfm}  -c^\text{dq}_{\text{cur},\tau}(\hat{\theta})\| \leq r_{\text{cur},\tau}  \}$. Letting $C^\text{dq}(\hat{\theta}) \!\coloneqq\! C^\text{dq}_\text{mod}(\hat{\theta}) \cap C^\text{dq}_{\text{cur},\tau_\text{ctr}}(\hat{\theta}) \cap C^\text{dq}_{\text{cur},\tau_\text{cyc}}(\hat{\theta})$, we obtain
%
\begin{align}\label{eq:proj:cart}
\Pi^{\text{dq}}_{\mc C^\text{dq}(\hat{\theta})} (\hat{v}^\text{dq}_\text{gfm}) \coloneqq &\argmin_{v^\text{dq}_\text{gfm} \in \mc C^\text{dq}(\hat{\theta})} \; (v^\text{d}_\text{gfm}-\hat{V})^2+\frac{w_\theta}{\hat{V}^2}(v^\text{q}_\text{gfm})^2
\end{align}
Notably, in Cartesian coordinates, the feasible set is represented by an intersection of disks and, hence, convex. An illustration of the constraint sets, dq-frame with angle $\hat{\theta}$, GFM reference voltage $\hat{v}^{\alpha\beta}_\text{gfm} = \mc R(\hat{\theta}) \hat{v}^\text{dq}_\text{gfm}$, and projected voltage $v^{\alpha\beta}_\text{gfm} = \mc R(\hat{\theta}) \Pi^{\text{dq}}_{\mc C^\text{dq}(\hat{\theta})} (\hat{v}^\text{dq}_\text{gfm})$ is shown in Fig.~\ref{fig:Calphabeta}. 
Note that changes to $\hat{v}^\text{d}_\text{gfm}$ modify the voltage magnitude, while small changes to $\hat{v}^\text{q}_\text{gfm}$ modify the voltage phase angle. Crucially, \eqref{eq:proj:cart} is a quadratically constrained quadratic program (QCQP) that can be solved efficiently, as shown next.
\begin{figure}[b!]
\centering
\includegraphics[width=0.9\columnwidth]{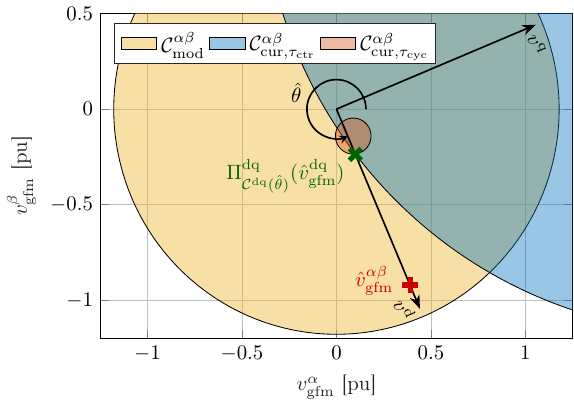}
\caption{Set of feasible voltages during a short-circuit fault: reference voltage $\hat{v}$ provided by GFM droop control (\textcolor{red!80!black}{$\Plus$}) in dq-frame with angle $\hat{\theta}$ and its projection onto the set of feasible voltages (\textcolor{green!40!black}{$\Xmark$}).\label{fig:Calphabeta}}
\end{figure}

\subsection{Efficient Projection in Cartesian Coordinates}\label{sec:splitting}
To make the control developed in the previous section practical, the projection operation needs to be solved on an embedded controller at high sampling rates. To this end, we next present an Alternating Direction Method of Multipliers (ADMM)~\cite[Sec.~3.1]{MAL-016} algorithm that is inspired by~\cite{osqp} and tailored to the projection \eqref{eq:proj:cart}. We begin by using $\omega_\text{dq}=\omega_0$ when computing the constraints \eqref{eq:cons.disc}. {This assumption enables offline computation of the matrices $M_\tau$ that define $\mc C^\text{dq}_{\text{cur},\tau}$ and is typically justified for the small frequency deviations observed in transmission systems. The impact of grid frequency deviations is illustrated in Sec.~\ref{subsec:lowfr}.} For brevity of the notation, let $C^\text{dq}_1 \coloneqq C^\text{dq}_{\text{cur},\tau_\text{ctr}}$, $C^\text{dq}_2 \coloneqq  \mc C^\text{dq}_{\text{cur},\tau_\text{cyc}}$, and $C^\text{dq}_1 \coloneqq \mc C^\text{dq}_\text{mod}$. Next, we introduce splitting variables $z=(z_1,z_2,z_3)$ with $z_n \in \mathbb{R}^2$ for $n \in \{1,2,3\}$ and rewrite \eqref{eq:proj:cart} as
\begin{subequations}\label{eq:proj:splitting}
  \begin{align}
&\min_{v^\text{dq}_\text{gfm},z_1,z_2,z_3} \; &\tfrac{1}{2}(v^\text{dq}_\text{gfm}&)^\mathsf{T} W v^\text{dq}_\text{gfm} + q^\mathsf{T} v^\text{dq}_\text{gfm}  \phantom{aaaa}\\
&\quad \text{s.t. }& v^\text{dq}_\text{gfm} &= z_n \qquad\quad\;\, \forall n \in \{1,2,3\} \label{eq:proj:splitting:splitting}\\ 
&\phantom{\quad \text{s.t. }}& z_n &\in \mc C^\text{dq}_n(\hat{\theta}) \qquad \forall n \in \{1,2,3\}
\end{align}
\end{subequations}
with $q=-W \hat{v}^\text{dq}_\text{gfm}$. Moreover, we introduce the penalty parameter $\rho \in \mathbb{R}_{>0}$ and (scaled) dual multipliers $y = (y_1,y_2,y_3)$ and $y_n \in \mathbb{R}^2$, $n \in \{1,2,3\}$ for the constraint \eqref{eq:proj:splitting:splitting}. Then, the augmented Lagrangian of \eqref{eq:proj:splitting} in scaled-dual form~\cite[Sec.~3.1.1]{MAL-016} is given by 
$\mathcal{L}_s = \tfrac{1}{2}(v^\text{dq}_\text{gfm})^\mathsf{T} W v^\text{dq}_\text{gfm} + q^\mathsf{T} v^\text{dq}_\text{gfm} + \sum_{n=1}^{3}  \frac{\rho}{2} \| x - z_n + y_n \|^2 + \mc I_{\mc C^\text{dq}_n(\hat{\theta})}(z_n)$ and the indicator function $\mc I_{\mc C^\text{dq}_n(\hat{\theta})}(z_n)$ (i.e., $\mc I_{\mc C^\text{dq}_n(\hat{\theta})}(z_n)=0$ if $z_n \in \mc C^\text{dq}_n(\hat{\theta})$ and $\mc I_{\mc C^\text{dq}_n(\hat{\theta})}(z_n)=\infty$ if $z_n \notin \mc C^\text{dq}_n(\hat{\theta})$). Moreover, the Euclidean projection of $\xi \in \mathbb{R}^2$ onto the disk $\mc C^\text{dq}_n(\hat{\theta})$ is denoted by
\begin{align}\label{eq:circlim}
\Pi_{\mc C^\text{dq}_n (\hat{\theta})}(\xi) \coloneqq \begin{cases}
\xi &\text{if } \xi \in \mc C^\text{dq}_n\\
\frac{\xi}{\|\xi\|}r_n + c^\text{dq}_n &\text{if } \xi \notin \mc C^\text{dq}_n.\\
\end{cases}
\end{align}
This corresponds to the well-known circular limiter, i.e., finding the nearest point to $\xi$ in the boundary of $\mc C^\text{dq}_n(\hat{\theta})$. Next, we recall $W=\diag\{1,w_\theta / \hat{V}^2\}$ and define $W_\rho \coloneqq W + 3\rho I_2$. Applying ADMM~\cite[Sec.~3.1]{MAL-016} with acceleration~\cite[Sec.~3.4.3]{MAL-016} to $\mathcal{L}_s$ results in the algorithm 
\begin{subequations}\label{eq:ADMM}
  \begin{align}
 v^\text{dq}_\text{gfm}(l\!+\!1) &\coloneqq W^{-1}_\rho \big(W \hat{v}^\text{dq}_\text{gfm} \!+\! \rho  \sum\nolimits_{n=1}^3  (z_n(l)\!-\!y_n(l))\big)\\
\!\tilde{v}^\text{dq}_\text{gfm}(l\!+\!1)  &\coloneqq v^\text{dq}_\text{gfm}(l\!+\!1) \!+\! (\alpha\!-\!1) (v^\text{dq}_\text{gfm}(l\!+\!1) \!-\! v^\text{dq}_\text{gfm}(l))\!\label{eq:ADMM:derivative}\\
z_n(l\!+\!1) &\coloneqq \Pi_{\mc C^\text{dq}_n(\hat{\theta})} \big(\tilde{v}^\text{dq}_\text{gfm}(l\!+\!1)+y_n(l)\big)\\
y_n(l\!+\!1) &\coloneqq y_n(l) + \tilde{v}^\text{dq}_\text{gfm}(l\!+\!1) - z_n(l+1)
\end{align}
\end{subequations}
with iteration index $l \in \mathbb{N}_0$ and $\alpha \in [1,2]$ typically selected as $\alpha \in [1.5,1.8]$ (see~\cite[Sec.~3.4.3]{MAL-016}). At every time step $k$, we initialize $z_n(0)=\hat{v}^\text{dq}_\text{gfm}(t_k)$ and $y_n(0)=\mathbbl{0}_2$ for all $n \in \{1,2,3\}$ and perform $n_\text{it} \in \mathbb{N}$ iterations of \eqref{eq:ADMM}. A block diagram of algorithm \eqref{eq:ADMM} is shown in Fig.~\ref{fig:ADMM}.

\begin{figure}[h!]
\centering
\includegraphics[width=1\columnwidth]{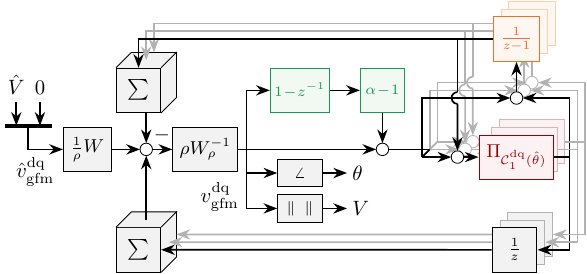}
\caption{ADMM algorithm with derivative for acceleration (green), integrals of the difference between the voltage $\tilde{v}^\text{dq}_\text{gfm}(l\!+\!1)$ and its limited version $z_n(l\!+\!1)$ (orange), and circular voltage limiter $\Pi_{\mc C^\text{dq}_n(\hat{\theta})}$ (red),  \label{fig:ADMM}}
\end{figure}

\subsection{Discussion and Interpretation}
At every iteration the algorithm \eqref{eq:ADMM} shown in Fig.~\ref{fig:ADMM} computes a weighted average of previous iterates and the voltage $\hat{v}^\text{dq}_\text{gfm}$ to obtain the voltage $v^\text{dq}_\text{gfm}$. In particular, if $\rho=1$, $w_\theta=\hat{V}^2$, and $y_n(l)=\mathbbl{0}_2$, then $v^\text{dq}_\text{gfm}(l\!+\!1)  = \frac{1}{4} (\hat{v}^\text{dq}_\text{gfm} + \sum_{n=1}^3 z_n)$. Next, the discrete-time proportional-derivative filter \eqref{eq:ADMM:derivative} is applied to accelerate the algorithm (i.e., take larger steps towards the optimizer). The resulting voltage $\tilde{v}^\text{dq}_\text{gfm}$ is passed through the circular limiters \eqref{eq:circlim} for each of the three constraints. Finally, the difference between $\tilde{v}^\text{dq}_\text{gfm}$ and the output of the circular limiters is integrated to obtain the dual multipliers $y_n$ that are fed back to decrease the difference between the voltage $\tilde{v}^\text{dq}_\text{gfm}(l+1)$ and its limited version $z_n(l+1)$ in the next iteration. A key advantage of this algorithm is that, besides the circular limiter, it only consists of few and simple operations. Notably, ADMM algorithms often exhibit rapid convergence to medium-accuracy solutions~\cite{osqp}. For the application at hand close approximations of the optimizer can be obtained in five to ten iterations (see Sec.~\ref{sec:scib}).

\section{Constraint-Aware GFM Droop Control}\label{sec:gfm} 

The discrete time constraint-aware GFM droop control with sampling time $\tau_\text{ctr} \in \mathbb{R}_{>0}$ and time $t_k \coloneqq k \tau_\text{ctr}$ for all $k\in \mathbb{N}_0$
is depicted in Fig.~\ref{fig:proj.gfm}. It consists of droop control with active and reactive power droop coefficients $m_p \in \mathbb{R}_{>0}$ and $m_q \in \mathbb{R}_{>0}$ that use low-pass filtered power measurements $P_\text{lp} \in \mathbb{R}$ and $Q_\text{lp} \in \mathbb{R}$ to compute a candidate GFM ac voltage with phase angle $\hat{\theta} \in [-\pi,\pi)$ and magnitude $\hat{V} \in \mathbb{R}_{>0}$. Subsequently, if the candidate ac voltage would result in a violation of current or modulation constraints, the controller computes an ac voltage with phase angle $\theta \in [-\pi,\pi)$ and magnitude $V \in \mathbb{R}_{>0}$ that is feasible for the constraints and minimizes the deviation from the candidate ac voltage $\hat{\theta} \in [-\pi,\pi)$ and magnitude $\hat{V} \in \mathbb{R}_{>0}$. We note that, for feasible candidate voltages (i.e., $\theta=\hat{\theta}$ and $V=\hat{V}$) this control reduces to standard GFM droop control with a low-pass filtered voltage magnitude dynamics. Finally, active $RC$ damping is used to suppress $LCL$ filter resonance~\cite{WBL+2015}.
\begin{figure}[b!]
\centering
\includegraphics[width=1\columnwidth]{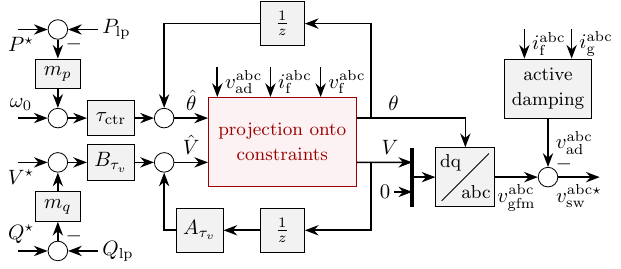}
\caption{The proposed constraint-aware GFM control leverages GFM droop control to compute the unconstrained candidate voltage $\hat{\theta}$, $\hat{V}$ that is limited to the set of voltages that do not result in violations of current or modulation limits to obtain the (constrained) GFM voltage reference (i.e., $\theta$ and $V$). Active damping is used to suppress $LCL$ filter resonance.\label{fig:proj.gfm}}
\end{figure}

\subsection{Active Damping of $LCL$ Filter Resonance}\label{sec:gfm:ad}
Beyond its fast acting current limiter, a benefit of the dual-loop architecture (see Fig.~\ref{fig:conv:dualloop}) is that the inner controls implicitly dampen $LCL$ filter resonance. At the same time, the tuning of dual-loop architecture can be challenging because it needs to account for $LCL$ filter resonance, current control and limiting, voltage control, and related aspects. In contrast, the constraint-aware GFM control shown in Fig.~\ref{fig:proj.gfm} uses separate dedicated controls for  damping of filter resonance and current limiting. In the context of current limiting and fault ride through, the resonance of the capacitor voltage upon fault inception, or more generally, when current limits are reached is a significant concern. Thus, for brevity of the presentation we focus on virtual $RC$ damping $v^\text{abc}_\text{ad}(s)=k_\text{rc} \frac{s}{s+\omega_\text{rc}} (i^\text{abc}_\text{g} - i^\text{abc}_\text{f})$ with gain $k_\text{rc}$ and cutoff frequency $\omega_\text{rc}$. The reader is referred to~\cite{WBL+2015} for an in-depth discussion of virtual $RC$ damping. 
We emphasize that similar methods for reshaping the converter output impedance, such as virtual impedance to mitigate subsynchronous oscillations or harmonic instability~\cite{WLB+2015}, can be readily combined with the proposed control.

\subsection{Constraint-Aware GFM Droop Control}
Given a time constant $\tau \in \mathbb{R}_{>0}$ we define $A_{\tau} \coloneqq e^{-\frac{\tau_\text{ctr}}{\tau}}$, $B_{\tau} \coloneqq 1-A_{\tau}$, and the first-order low-pass filter 
\begin{subequations}
\begin{align}
 P_\text{lp}(t_k) &= A_{\tau_\text{lp}} P_\text{lp}(t_{k-1}) + B_{\tau_\text{lp}}  P_\text{f}(t_k),\\
 Q_\text{lp}(t_k) &= A_{\tau_\text{lp}}  Q_\text{lp}(t_{k-1}) + B_{\tau_\text{lp}}  Q_\text{f}(t_k)
\end{align}
\end{subequations}
with time constant $\tau_\text{lp} \in \mathbb{R}_{>0}$ used to, e.g., eliminate measurement noise and unwanted harmonics~\cite{HNL+2015} or to emulate an inertia response~\cite{SGR+2013}. Recall that the unconstrained GFM droop control \eqref{eq:droop:disc} is used to compute a candidate voltage phase angle and magnitude $(\hat{\theta}(t_k),\hat{V}(t_k))$. Subsequently, we employ the projection developed in Sec.~\ref{sec:proj} to enforce the constraints. The projection in polar coordinates results in
\begin{align}\label{eq:droop:disc:update}
(\theta(t_k),V(t_k)) \coloneqq \Pi^\text{pol}_{\mc C^{\alpha\beta}(t_k)}(\hat{\theta}(t_k),\hat{V}(t_k)),
\end{align}
where $\mc C^{\alpha\beta}(t_k)$ has been constructed using the frequency $\omega_\text{dq}=\omega_\text{dr}$, $v^\text{abc}_\text{ad}(t_k)$ computed by virtual RC damping, and measurements of $v^\text{abc}_\text{f}(t_k)$, and $i^\text{abc}_\text{f}(t_k)$. Thus, if the unconstrained GFM voltage $(\hat{\theta}(t_k),\hat{V}(t_k))$ violates modulation limits and/or is predicted to violate current limits, it is replaced with the closest feasible voltage  $(\theta(t_k),V(t_k))$ as parameterized by the weight $w_\theta \in \mathbb{R}_{>0}$ in \eqref{eq:proj:polar}.

{Notably, the projection \eqref{eq:droop:disc:update} in polar coordinates is often too computationally expensive for implementations on commonly used microcontrollers. Thus, for typical application scenarios, the projection in Cartesian coordinates is recommended. In particular, letting $\hat{v}^\text{dq}_\text{gfm}(t_k)=(\hat{V}(t_k),0)$, the projection $v^\text{dq}_\text{gfm} (t_k) \coloneqq \Pi^\text{dq}_{\mc C^\text{dq}(\hat{\theta})}(\hat{v}^\text{dq}_\text{gfm}(t_k))$ is computed in the dq frame with angle $\hat{\theta}$, and the voltage phase angle and magnitude are recovered} using $\theta(t_k) \coloneqq \|v^\text{dq}_\text{gfm} (t_k)\|$ and $,V(t_k) \coloneqq \angle\, v^\text{dq}_\text{gfm} (t_k)$.
%

{\subsection{Computational Complexity}}\label{subsec:complexity}
{Before investigating embedded controller implementations, computation times of the projection in polar coordinates \eqref{eq:proj:polar} and Cartesian coordinates using \eqref{eq:ADMM} are compared. This comparison is performed in simulation for the system shown in Fig.~\ref{fig:caseib} because real-time computation of \eqref{eq:proj:polar} using general purpose nonlinear programming solvers is beyond the capabilities of typical embedded controllers. Computation times in Simulink on a MacBook Air M2 are provided in Tab.~\ref{tab:comptime}.
\begin{table}[b!!]
\begin{center}
\caption{{Computation time for projecting GFM voltage references} \label{tab:comptime}}
\begin{tabular}{ ccccccc }       
projection & algorithm & $\rho$ & $n_\text{it}$ & min & mean & max \\ \midrule
$\Pi^\text{pol}_{\mc C^{\alpha\beta}}$ & SQP (fmincon) & - & - & $47$~$\mu$s & $83.03$~$\mu$s  & $671$~$\mu$s\\ 
$\Pi^\text{dq}_{\mc C^\text{dq}}$ & ADMM \eqref{eq:ADMM} & $1$ & $10$ & $3$~$\mu$s & $3.66$~$\mu$s  & $10$~$\mu$s\\ 
$\Pi^\text{dq}_{\mc C^\text{dq}}$ & ADMM \eqref{eq:ADMM} & $5$ & $5$ & $1$~$\mu$s & $1.78$~$\mu$s  & $3$~$\mu$s\\
        \bottomrule
\end{tabular}
\end{center}
\end{table}
Next, we focus on the computational complexity of constraint-aware droop control with projection algorithm~\eqref{eq:ADMM}. Table~\ref{tab:comp} summarizes the instructions required to implement constraint-aware droop control, droop control with inner reference current limiter~\cite{FW2022}, dVOC with inner reference current limiting and anti-windup~\cite{ALJ+2022}, and droop control with hybrid threshold virtual impedance (TVI)~\cite{ZBS+2023}. To focus solely on the GFM control, the instruction count does not include common supporting tasks (e.g., I/O, filtering of measurements, active damping, and modulation algorithms). 
Using $n_\text{it}\!=\!5$ for comparison (see Sec.~\ref{sec:scib}), constraint-aware droop control has approximately $2.5$ times to $3.5$ times the computational complexity of widely studied GFM controls. } 
\begin{table}[b!!]{
  \setlength{\tabcolsep}{4pt}
  \centering
  \begin{threeparttable}
    \caption{Instruction count and processor cycles per instruction}
    \label{tab:comp}
    \begin{tabular}{@{}lcccc|cc@{}}
      & \multicolumn{4}{c}{GFM control instruction count} & \multicolumn{2}{c}{cycles / instruction} \\
      \cmidrule(r){2-5}\cmidrule(l){6-7}
      & Fig.~\ref{fig:proj.gfm} & \cite{FW2022} & \cite{ALJ+2022} & \cite{ZBS+2023}
      & TI~\cite{TI_TMS320F2837xD_SPRS880P_2024} & NXP~\cite{Freescale_E600CORERM_2006} \\
      \midrule
      comparison      & $3 n_{\text{it}}$       & $1$  & $1$  & $4$  & $2$   & $5$ \\      
      addition        & $36 n_{\text{it}} + 6$  & $40$ & $46$ & $52$ & $2.5$ & $9$ \\
      multiplication  & $10 n_{\text{it}} + 6$  & $41$ & $64$ & $70$ & $2.5$ & $9$ \\
      division\tnote{1} & $3 n_{\text{it}}$     & $1$  & $1$  & $0$  & $6.5$ & $39$ \\
      square root     & $3 n_{\text{it}} + 1$   & $1$  & $1$  & $2$  & $6.5$ & $39$ \\
      trigonometric   & $3$           & $2$  & $2$  & $2$  & $8.5$ & $1000$ \\
      \midrule
      total instructions          & $55 n_{\text{it}} + 16$ & $86$ & $115$& $130$& \multicolumn{2}{c}{} \\
      \bottomrule
    \end{tabular}
    \begin{tablenotes}[flushleft]
      \item[1] Divisions by constants are counted as multiplications.
    \end{tablenotes}
  \end{threeparttable}}
\end{table}

\begin{figure}[b!]
  \centering
\includegraphics[width=1\columnwidth]{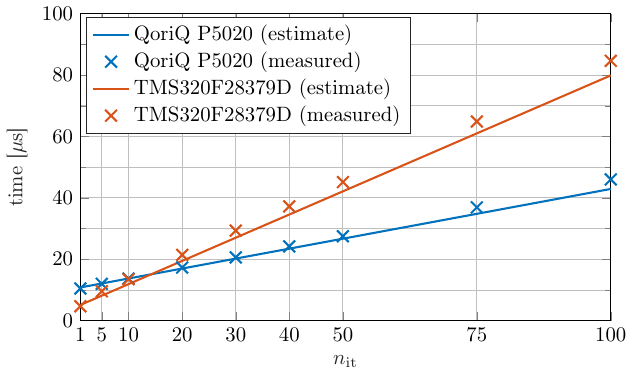}
\caption{{Estimates and actual computation time of constraint-aware GFM control on a TI TMS320F28379D and NXP QoriQ P5020.} \label{fig:cmptime}}
\end{figure}

The resulting computation times for two embedded controller platforms (TI TMS320F28379D and NXP QoriQ P5020) are estimated using the cycles per instruction estimates listed in Tab.~\ref{tab:comp}. Moreover, the controller has been deployed to both platforms using Simulink Coder to measure computation time including supporting tasks (e.g., I/O, filtering measurements, computing duty cycles). Actual computation times and analytical estimates (with the supporting task execution time added)  are shown in Fig.~\ref{fig:cmptime}. Different performance characteristics emerge due to the lower clock speed of the TMS320F28379D and lack of hardware support for trigonometric functions on the QoriQ P5020. Overall, sampling rates up to $50$~kHz are viable for $n_{\text{it}}=5$ on both platforms.

{\subsection{Implications for Ancillary Services and Grid Codes}}\label{sec:gfm:disc} 
Substituting the unconstrained GFM phase angle update \eqref{eq:droop:disc:theta} and $\theta(t_k) = \theta(t_{k-1}) + \tau_\text{ctr} \omega(t_k)$ into the projection \eqref{eq:proj:polar} the controller update \eqref{eq:droop:disc:update} can be rewritten as 
$(\theta(t_k),V(t_k)) = (\theta(t_{k-1})+\tau_\text{ctr} \omega(t_k)),V(t_k)$ with frequency $\omega(t_k)$ given by
\begin{align*}
(\omega(t_k),V(t_k)) \!=\! &\argmin_{\tilde{\omega},\tilde{V}} \; w_\omega (\tilde{\omega} - \omega_\text{dr}(t_k))^2+ (\tilde{V} \!-\! \hat{V}(t_k))^2 \\
&\; \text{s.t. }  \; \mc R_{\text{2D}}(\theta(t_{k-1}) + \tau_\text{ctr} \tilde{\omega}) \begin{bmatrix} V \\ 0 \end{bmatrix} \in \mc C^{\alpha\beta}(t_k)
\end{align*}
and weight $w_\omega = w_\theta \tau_\text{ctr} \in \mathbb{R}_{>0}$. The constraint $\mc C^{\alpha\beta} =  \mc C^{\alpha\beta}_{\text{cur},\tau_\text{ctr}} \cap \mc C^{\alpha\beta}_{\text{cur},\tau_\text{cyc}}$ can be interpreted as enforcing the current limit at the next controller time step and a stabilizing terminal condition~\cite[Ch.~5]{GP2016} on the voltage at $t_k+\tau_\text{cyc}$. Broadly speaking, this constraint ensures existence of a feasible sinusoidal voltage reference with frequency $\omega_\text{dr}$, phase angle $\omega(t_k)$, and magnitude $V(t_k)$ for $t \in [t_k,t_k+\tau_\text{cyc}]$. In abstraction this encodes the requirement to {present} a steady-state voltage {on subtransient timescales (e.g., $1$ to $3$ cycles)} after an event leading to overcurrent (e.g., short-circuit fault) that is currently under discussion for GFM {requirements (see, e.g., 
\cite[p.~6]{UNIFI_Specs_V2_2024}, \cite[p.~21]{ENTSOE_GFM_PPM_Interim_2024})}. In particular, omitting the constraint $v^{\alpha\beta}_\text{gfm} \in \mc C^{\alpha\beta}_{\text{cur},\tau_\text{cyc}}$ can result in tracking the unconstrained droop control \eqref{eq:droop:cont} for too long and reaching a state at which no feasible solution with frequency reasonably close to $\omega_\text{dr}$ exists. Ultimately, this can lead to pathological responses with very large frequency deviations during faults.  

If the unconstrained droop control \eqref{eq:droop:cont} is not feasible, the weight $w_\omega$ determines the trade-off between maintaining a frequency close to the unconstrained droop control \eqref{eq:droop:cont:freq} or a voltage close to the unconstrained droop control \eqref{eq:droop:disc:volt}. 
This interpretation allows to directly quantify and specify the trade-off between frequency and voltage support without using active or reactive current (or power) as proxies. In particular, specifications for voltage support and low voltage ride through for grid-connected converters are commonly formulated as specifications on the reactive current. However, the amount of reactive current needed to provide a desired level of voltage support is ultimately highly dependent on site-specific and time-varying grid parameters (e.g., short-circuit ratio, $\ell/r$ ratio). In contrast, our conceptual framework allows formulating specifications directly on the converter voltage without knowledge of site specific and time-varying grid parameters.

Finally, the proposed control provides a pathway to requirements that only allow momentary cessation or converter tripping if $\mc C^{\alpha\beta}(t_k)$ is empty, i.e., if the projection \eqref{eq:proj:cart} is infeasible (see Sec.~\ref{subsec:feasibleset}).

\subsection{{Controller Parametrization and Tuning}}\label{subsec:tuning}
The tuning of conventional inner vector controls is highly sensitive to grid strength~\cite{QGC+18,SG+20}. Similarly, tuning of threshold virtual impedance limiters is highly sensitive to grid impedance (see \cite[Sec.~6.1]{DPD+2018} and the references therein). 

In contrast, most tuning parameters of constraint-aware droop control are specified a priori. A notable exception is the voltage control time constant $\tau_\text{v} \in \mathbb{R}_{>0}$ in \eqref{eq:droop:disc} that needs to account for grid strength to avoid oscillations that may occur under low SCR if $\tau_\text{v} \in \mathbb{R}_{>0}$ is too small (e.g., if voltage recovery after a fault is too fast). In contrast, $m_p  \in \mathbb{R}_{>0}$ and $m_q  \in \mathbb{R}_{>0}$ are typically given (e.g., by markets or grid codes). The filter time constant $\tau_\text{lp} \in \mathbb{R}_{>0}$ is selected to eliminate measurement noise and unwanted harmonics~\cite{HNL+2015} and, e.g., can be used to emulate an inertia response~\cite{SGR+2013} if required. Moreover, the weight $w_\omega \in \mathbb{R}_{>0}$ expected to be selected in accordance with grid codes or system operator specifications and the remaining parameters of the projection \eqref{eq:proj:cart} are converter parameters.

Efficiently solving the projection may require tuning parameters of the algorithm \eqref{eq:ADMM}. Empirically, the ADMM literature has shown $\alpha = 1.6$ to perform well across vastly different problems~\cite{osqp}. The step size $\rho$ and number of iterations $n_\text{it}$ depends on the required solution tolerance. Notably, smaller filter inductances may require more accurate solutions and hence require smaller step size $\rho$ and an increased number of iterations $n_\text{it}$. However, for typical two-level VSCs filter inductances (e.g., $5$\% to $12$\%), $\rho \in [1,5]$ and five to ten iterations appear sufficient. Finally, for typical sampling rates (i.e., $1$~kHz to $10$~kHz), the impact on $\rho$  and $n_\text{it}$ is expected to be minimal because the one-cycle constraint is typically binding, but the sampling rate only affects the one-step constraint.

\section{{Comparison of projection algorithms}}\label{sec:scib}
To benchmark and compare the performance of the projection \eqref{eq:proj:polar} in polar coordinates and the projection \eqref{eq:proj:cart} in Cartesian coordinates computed using the algorithm \eqref{eq:ADMM}, a single converter connected to an infinite bus and a dc voltage source  (see Fig.~\ref{fig:caseib}) is simulated with the breaker $s_\text{vsc}$ closed. 

\begin{figure}[b!!]
  \centering
\includegraphics[trim={0.55cm 0 0 0},width=0.95\columnwidth,clip]{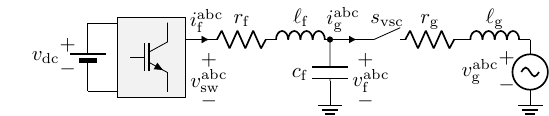}
\caption{Test system containing a single converter connected to an infinite bus.\label{fig:caseib}}
\end{figure}

The short-circuit ratio (SCR) at the converter bus is $7.5$ {and the remaining parameters are summarized in Tab.~\ref{tab:vsc}. To focus on the comparison between the two problem formulations, an averaged model of a three-phase two-level voltage source converter is used and computation delays are neglected in the simulation.} 
\begin{table}[t!]
\begin{center}
\caption{Parameters for the single converter infinite bus system\label{tab:vsc}}
\begin{tabular}{ cccccc }          \toprule 
  \multicolumn{2}{c}{Grid \& dc source} & \multicolumn{2}{c}{VSC} & \multicolumn{2}{c}{GFM control} \\ \toprule 
$V_\text{base}$ & $208$~V & $\ell_\text{f}$ & $0.075$~pu &  $m_p$, $m_q$ & $3$\% \\
$S_\text{base}$ & $2$~kW &  $r_\text{f}$ & $0.0076$~pu & $\tau_\text{v}$ &  $8$~ms \\ 
$f_\text{base}$ & $60$~Hz & $c_\text{f}$ & $0.09$~pu & $\tau_\text{lp}$ & $5.3$~ms \\
{$v_\text{dc}$} & {$400$~V} &  & & $\tau_\text{ctr}$ & $0.1$~ms \\[0.2em] \toprule 
  \multicolumn{2}{c}{Constraints} & \multicolumn{2}{c}{Projection}  & \multicolumn{2}{c}{virtual RC}\\ \toprule
  $V_{\max}$ & $1.178$~pu & $w^\text{pu}_\omega$ & $0.5$ & $k_\text{rc}$ & $0.1$~pu \\
  $i_{\max}$  & $1.2$~pu & $\tau_\text{cyc}$ & $20$~ms & $\omega_\text{rc}$ & $10^4$~rad/s \\  \bottomrule
\end{tabular}
\end{center}
\end{table} 
Due to the infinite bus, the weight $w_\theta$ does not have a significant impact, {and $w_\theta = \frac{w^\text{pu}_{\omega}}{\omega_0 \tau_\text{ctr}}$ with $w^\text{pu}_\omega=1/2$ is selected}. All simulations have been performed using Simscape Electrical with simulation step size $\tau_\text{sim}=1$~$\mu$s. 

\subsection{Symmetric Short-Circuit Fault}\label{subsec:symshort}
A symmetric short-circuit fault is simulated by reducing the infinite bus voltage $v^\text{abc}_\text{g}$ to zero. The projection \eqref{eq:proj:polar} in polar coordinates is computed using fmincon with the sequential quadratic programming (SQP) solver. The projection \eqref{eq:proj:cart} in Cartesian coordinates is computed using the ADMM algorithm \eqref{eq:ADMM} with $\alpha=1.6$. The resulting filter current magnitude, filter voltage magnitude, {and VSC voltage phase angle $\theta_\text{gfm} = \angle v^{\alpha\beta}_\text{sw}$,} are shown in Fig.~\ref{fig:IthV}. 
\begin{figure}[b!]
  \centering
\includegraphics[width=1\columnwidth]{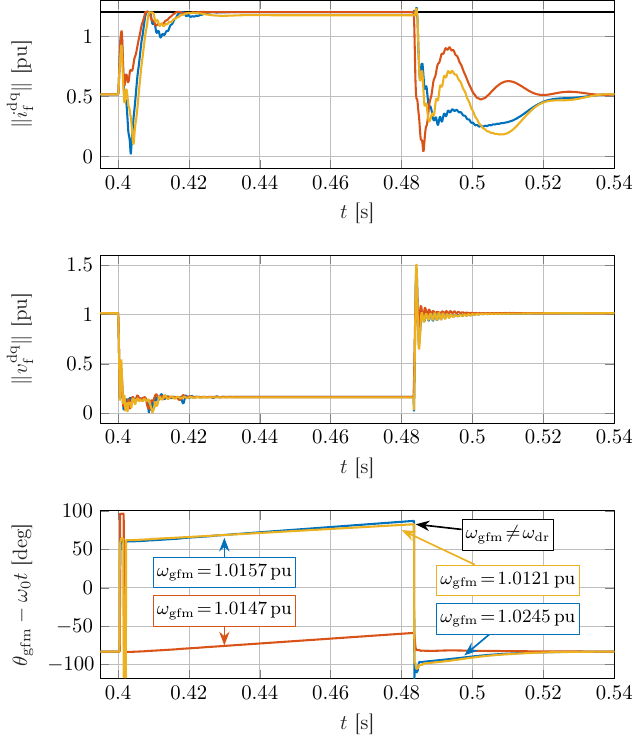}
\caption{Simulation results for a symmetric short-circuit fault applied at the infinite bus in Fig.~\ref{fig:caseib}: $\Pi^\text{pol}_{\mc C^{\alpha\beta}}$ (\protect\blueline), $\Pi^\text{dq}_{\mc C^\text{dq}}$ with $\rho=1$ and $n_\text{it}=10$ (\protect\redline), $\Pi^\text{dq}_{\mc C^\text{dq}}$ with $\rho=5$ and $n_\text{it}=5$ (\protect\yellowline), and current limits (\protect\blackline).
\label{fig:IthV}}
\end{figure}

It can be seen that the current magnitude stays within the prescribed limits and reaches the current limit within one cycle. Moreover, the filter capacitor voltage magnitude settles into steady-state within approximately one cycle and only exceeds the modulation limit during a brief transient after fault clearing. {Please note that, $\theta_\text{gfm} \neq \angle v^{\alpha\beta}_\text{sw}$ when inner (see Fig.~\ref{fig:conv:dualloop}) or parallel limiters (see Fig.~\ref{fig:conv:tvi}) are active. In this case, the actual ac voltage phase angle $\angle v^{\alpha\beta}_\text{sw}$ typically also changes rapidly upon fault inception or clearing, while $\theta_\text{gfm}$ commonly shown in the literature may not.}  {The GFM controller frequency settles to a steady-state within one to two cycles after fault inception and the steady-state frequencies match the reference frequency $\omega_\text{dr}=1.015$ of the droop controller~\eqref{eq:droop:cont:freq} to within $0.04\%$ ($\rho=5$) to $0.29\%$ ($\rho=1$).}

Overall, the two algorithms exhibit very similar responses at fault inception and during the fault. More pronounced differences can be observed after clearing the fault. However, even during the fault recovery the responses exhibit approximately the same settling time and transients are comparable. We note that the projection in {Cartesian coordinates} solved using the ADMM algorithm~\eqref{eq:ADMM} with ten iterations and step size $\rho=1$ results in comparable responses to five iterations at an increased step size of $\rho=5$. With fewer than five iterations, the results degrade significantly. 

Waveforms obtained using the projection $\Pi^\text{dq}_{\mc C^\text{dq}}$ in Cartesian coordinates $n_\text{it}=5$ iterations of the ADMM algorithm~\eqref{eq:ADMM} with step size $\rho=5$ are shown in Fig.~\ref{fig:ib:phase}. 
\begin{figure}[b!]
  \centering
\includegraphics[width=1\columnwidth]{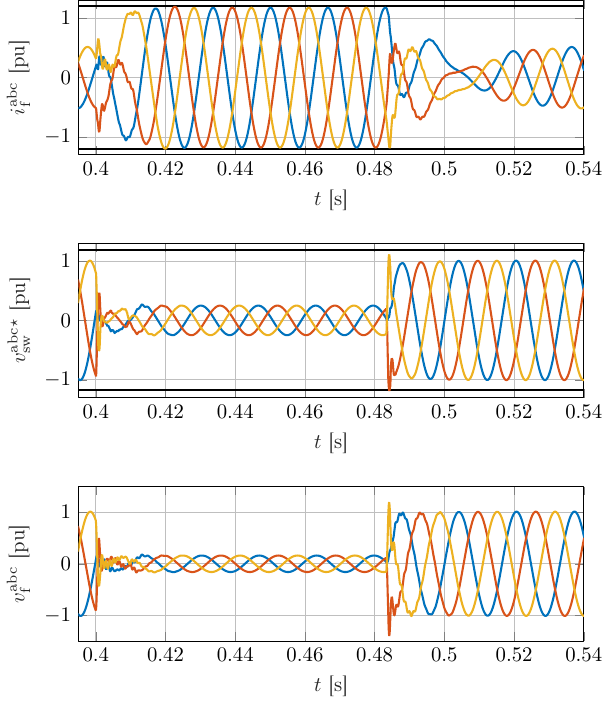}
\caption{{Filter currents, reference voltages, and filter capacitor voltages,}  for a symmetric short-circuit fault applied at the infinite bus in Fig.~\ref{fig:caseib}: phase $a$ (\protect\blueline), phase $b$ (\protect\redline), phase $c$ (\protect\yellowline), and current limits (\protect\blackline).\label{fig:ib:phase}}
\end{figure}
It can be seen that the converter filter voltages settle to balanced sinusoidal voltages within approximately one cycle both during fault inception and fault recovery. {Brief $LC(L)$ filter resonance is observed for $v^\text{abc}_\text{f}$ and $\| v^\text{dq}_\text{f}\|$ (see Fig.~\ref{fig:IthV}) upon fault inception. In particular, current limiting often briefly reduce the effectiveness of the virtual RC damping~\cite{WBL+2015} used to suppress capacitor voltage resonance. Capacitor voltage resonance may be reduced under current and modulation limits at the expense of increased losses by using a damping resistor $r_\text{D}$ (see Fig.~\ref{fig:testbed:schematic} and Sec.~\ref{sec:experiment}). Moreover, the brief voltage transient upon fault clearing is well within the range of typical transient overvoltage ride-through requirements~\cite[Sec.~7.2.3]{IEEE2800-2022}.} At all other times, the assumptions of Prop.~\ref{prop:nonempty} are satisfied and feasibility of the projection is ensured.

Finally, the settling time of the current is mostly dictated by the dynamics of the grid inductance $\ell_\text{g}$ and the medium coupling strength (i.e., SCR of $7.5$) used for this case study, the settling time of the current to steady-state is approximately one cycle upon fault inception and approximately two cycles during fault recovery. Overall, we observe that constraint-aware GFM control results in no significant constraint violation, exhibits rapid current rise time as expected of GFM controls, and rapid resynchronization and fault recovery.
\subsection{Robustness to Parameter Uncertainty}\label{subsec:robustness}
Finally, we investigate the robustness to parameter uncertainty. To ensure that the solution accuracy does not impact the result, we consider the projection $\Pi^\text{dq}_{\mc C^\text{dq}}$ in Cartesian coordinates solved by the ADMM algorithm \eqref{eq:ADMM} with $\rho=1$ and $n_\text{it}=10$. Simulations for a symmetric short-circuit fault at the infinite bus have been conducted using values $\ell^\text{ctr}_\text{f} \in [0.9\,\ell_\text{f},1.05\,\ell_\text{f}] $ and $r^\text{ctr}_\text{f} \in [0,2\,r_\text{f}]$ for controller design. The resulting peak and steady-state overcurrent are shown in Fig.~\ref{fig:ib:rob}.
\begin{figure}[b!!]  
  \centering 
\includegraphics[width=0.99\columnwidth]{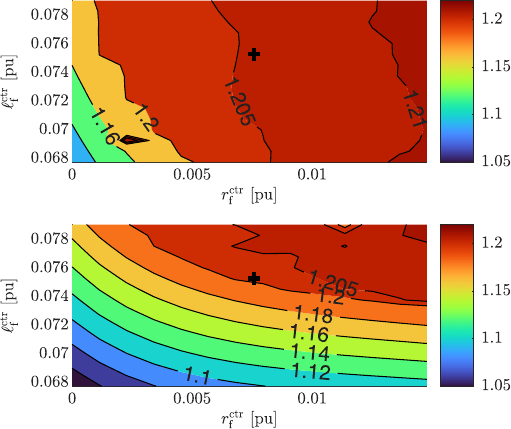}
\caption{Peak (top) and steady-state (bottom) current during a symmetric short-circuit fault when using $\ell^\text{ctr}_\text{f}$ and $r^\text{ctr}_\text{f}$ in the controller instead of the true values $\ell_\text{f}$ and $r_\text{f}$ ($\Plus$).\label{fig:ib:rob}}
\end{figure}


It can be seen that $\ell^\text{ctr}_\text{f} < \ell_\text{f}$ and $r^\text{ctr}_\text{f} < r_\text{f}$ results in underutilization of the converter current capability because the size of $\mc C^\text{dq}$ is reduced. While $\ell^\text{ctr}_\text{f} > \ell_\text{f}$ and $r^\text{ctr}_\text{f} > r_\text{f}$ results in constraint violations, they are not significant when $\ell^\text{ctr}_\text{f} - \ell_\text{f}$ and $r^\text{ctr}_\text{f} - r_\text{f}$ is within typical tolerances. We observe that the fault response exhibits oscillations when $\ell_\text{f}^\text{ctr} > 1.06\, \ell_\text{f}$.

\subsection{{Low-frequency Ride-through}}\label{subsec:lowfr}
{To investigate the impact of grid frequency deviations, a $5\%$ drop of the infinite bus frequency is simulated for the system shown in Fig.~\ref{fig:caseib}. This disturbance reflects typical low-frequency ride-through requirements (see, e.g., \cite[Sec.~7.3.2.1]{IEEE2800-2022}). Simulations for the projection $\Pi^\text{pol}_{\mc C^{\alpha\beta}}$ with $\omega_{\text{dq}}=\omega_\text{dr}$ and the projection $\Pi^\text{dq}_{\mc C^\text{dq}}$ computed using \eqref{eq:ADMM} with $\rho=5$, $n_\text{it}=5$, and $\omega_{\text{dq}}=\omega_0$ are shown in Fig.~\ref{fig:drop}.}
\begin{figure}[t!!!]
  \centering
\includegraphics[width=0.99\columnwidth]{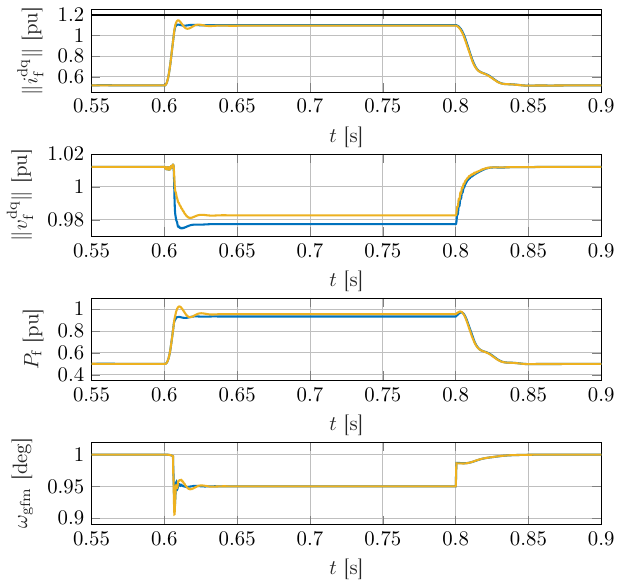}
\caption{{Simulation results for a $5$~\% drop of the grid frequency in Fig.~\ref{fig:caseib}: $\Pi^\text{pol}_{\mc C^{\alpha\beta}}$ with $\omega_{\text{dq}}=\omega_\text{dr}$ (\protect\blueline), $\Pi^\text{dq}_{\mc C^\text{dq}}$ with $\rho=5$, $n_\text{it}=5$, and $\omega_{\text{dq}}=\omega_0$ (\protect\yellowline), and current limit (\protect\blackline). The GFM controller frequency $\omega_\text{gfm}$ is obtained using finite differencing of $\theta(t_k)$.}\label{fig:drop}}
\end{figure}
{The pre-contingency operating point is $P=0.5$~pu and the active power injection of $P=1.5$~pu expected under a $5\%$ frequency drop would exceed the converter current limit at nominal voltage. Instead, the limiter activates and limits the current to approximately $1.1$~pu. Overall, it can be seen that (i) the response is not significantly impacted by the choice of frequency used in the predictor and the projection algorithm, and (ii) operating within the one-step and one-cycle constraint results in preserving approximately $10\%$ headroom in the current. Moreover, the voltage magnitude is impacted by the limiter, but remains within an acceptable range for a $208$~V system. The voltage profile can be improved by increasing the priority of tracking the GFM voltage magnitude under constraints (i.e., by decreasing $w^\text{pu}_\omega$).}

\section{Comparison with Current Reference Limiting and Variable Virtual Impedance}\label{sec:comp}
Next, a simulation of the system in Fig.~\ref{fig:caseib} is used to compare the proposed method with widely used implementations of inner current limiting (see Fig.~\ref{fig:conv:dualloop}) and parallel current limiting (see Fig.~\ref{fig:conv:tvi}). To this end, cascaded vector current and voltage proportional-integral control~\cite{QGC+18} are combined with either a circular current reference limiter \cite[Fig.~2b]{WWZ2024} or variable virtual impedance~\cite[Fig.~15]{QWW+2023}, \cite[Sec.~II-B]{ZBS+2023}. The proportional and integral gain of the current controller are $1$~pu and $0.24$~pu, respectively. The proportional and integral gain of the voltage controller are $0.55$~pu and $0.23$~pu, respectively. The variable virtual impedance current limiter is parameterized as in \cite{ZBS+2023} with steady-state virtual reactance-to-resistance ratio $\rho_\frac{X}{R}=5$, transient reactance-to-resistance ratio of $0.8$, activation threshold of $1$~pu, and high pass filter cutoff frequency of $1000~\mathrm{rad/s}$. The gain of the variable virtual impedance is computed such that the current is limited to $i_{\max}$ during a bolted fault at the converter terminal~\cite[Eq. (2)]{ZBS+2023}. The maximum current of the circular current reference limiter is set to $1.2$~pu. 

\subsection{Synchronization, frequency drop, and short-circuit fault}\label{sec:comp:timedomain}
Simulation results for constraint-aware GFM control, variable virtual impedance, and the circular current reference limiter are shown in Fig.~\ref{fig:TVI}. Initially, the breaker $s_\text{vsc}$ is open. The initial difference between the infinite bus phase angle $\theta_g = \angle v_\text{g}$ and GFM controller phase angle $\theta$ is $180^\circ$. 
\begin{figure}[b!] 
  \centering
\includegraphics[width=1\columnwidth]{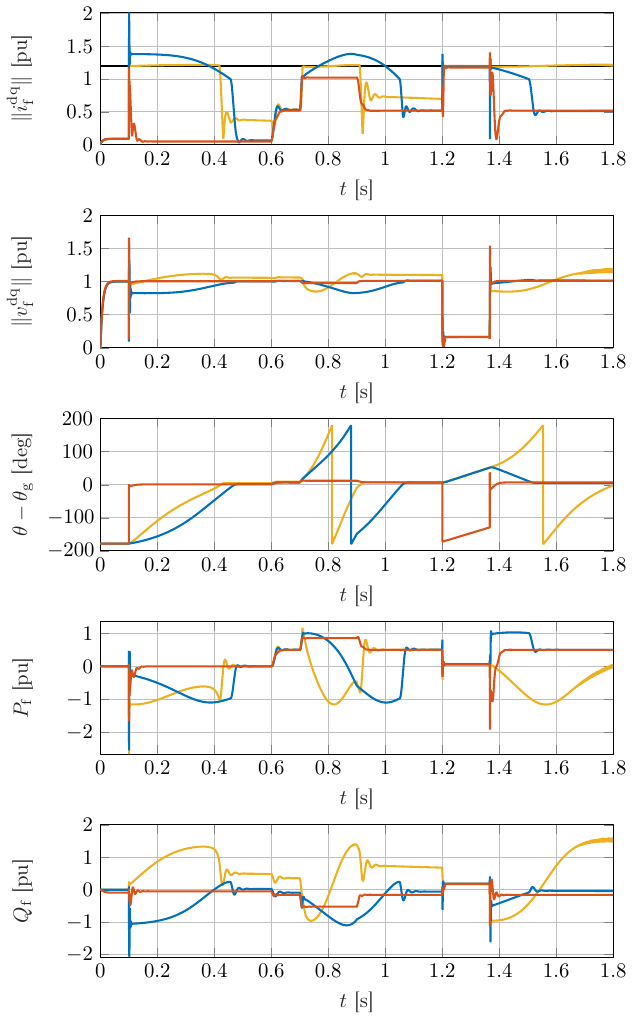}
\caption{Simulation results for the system in Fig.~\ref{fig:caseib}: Reference current limiting (\protect\yellowline), variable virtual impedance (\protect\blueline), constraint-aware GFM control $\Pi^\text{dq}_{\mc C^\text{dq}}$ with $\rho=5$ and $n_\text{it}=5$ (\protect\redline), and current limits (\protect\blackline).\label{fig:TVI}}
\end{figure}

At $t=0.1$~s the breaker $s_\text{vsc}$ closes. Constraint-aware GFM control immediately synchronizes and does not violate the current limit. Variable virtual impedance current limiting violates the current limit, results in a prolonged synchronization transient, and exhibits significant undervoltage. While current reference limiting does not violate the current limit it does not fully resynchronize before the next event and exhibits significant overvoltage. 

At $t=0.6$~s, the power setpoint is updated to $P^\star_\text{f}=0.5$~pu. From $t=0.7$~s to $t=0.9$~s, a $5$~\% drop of the grid frequency is simulated. Constraint-aware GFM control synchronizes to the grid frequency, provides the expected active power response while maintaining headroom in the current, and immediately returns to the pre-contingency operating point after the grid frequency returns to the nominal value. In contrast, variable virtual impedance results in a loss of synchronization of the droop control phase angle $\theta$, a violation of the current limit, and a severe drop in the filter voltage $v_f$. In other words, the combination of GFM droop control and a variable virtual impedance does not provide the expected frequency and voltage support. Notably, in this scenario, providing GFM functions while remaining within the converter constraints requires modifying the frequency of the GFM controller, but variable virtual impedance can only modify the GFM voltage $v^{\text{abc}}_\text{gfm}$ in proportion to the overcurrent. While current reference limiting does not violate the current limit, it also results in a loss of synchronization, significant under and overvoltage, and a prolonged resynchronization transient.

Finally, from $t=1.2$~s to $t=1.3667$~s, a symmetric short-circuit fault is simulated by reducing $v^\text{abc}_\text{g}$ to zero. While constraint-aware GFM control resynchronizes immediately upon fault clearing, a prolonged resynchronization transient is observed for the virtual variable impedance current limiter. While current reference limiting does not violate the current limit, the fault duration exceeds its critical clearing time and the GFM control with current reference limiting permanently loses synchronization and fails to resynchronize after the fault is cleared.

\subsection{Voltage support during grid faults}\label{sec:comp:mag}
Next, we compare the voltage support capabilities of constraint-aware GFM control and variable virtual impedance during a sag of the grid voltage to $0.2$~pu. While the optimal tuning of the variable virtual impedance requires knowledge of the grid reactance-to-resistance ratio~\cite[Sec.~V]{GD2019},\cite[Sec.~III]{WWZ2024}, neither the design nor the tuning of constraint-aware GFM control require this information. To empirically demonstrate the implications of this observation for voltage support, we compare constraint-aware GFM control, variable virtual impedance with the optimal tuning $\rho_\frac{X}{R}=\omega_0 \ell_\text{g}/r_\text{g}$, and variable virtual impedance with worst-case tuning $\rho_\frac{X}{R}=0$. Simulation results for $\omega_0 \ell_\text{g}/r_\text{g}$ ranging from $2$ to $7$ are summarized in Fig.~\ref{fig:mag}.
\begin{figure}[h!]
  \centering
\includegraphics[width=1\columnwidth]{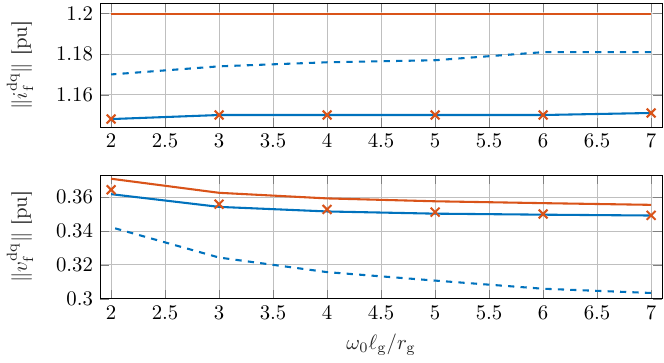}
\caption{Current and voltage magnitude under a grid voltage sag to $0.2$~pu as a function of grid reactance-to-resistance ratio:  variable virtual impedance with $\rho_\frac{X}{R}=\omega_0 \ell_\text{g}/r_\text{g}$ (\protect\blueline), variable virtual impedance with $\rho_\frac{X}{R}=0$ (\protect\dashedblueline), constraint-aware GFM control with $i_{\max}=1.2$~pu (\protect\redline) and $i_{\max}$ matched to variable virtual impedance with $\rho_\frac{X}{R}=\omega_0 \ell_\text{g}/r_\text{g}$  (\protect\redmark).\label{fig:mag}} 
\end{figure} 
Because the variable virtual impedance is parameterized to limit the current to $i_{\max}=1.2$ during the worst-case fault at the converter terminal, it is not utilizing the full maximum current during the voltage sag behind a grid impedance. As predicted by~\cite[Sec.~V]{GD2019},\cite[Sec.~III]{WWZ2024}, the variable virtual impedance with optimal tuning outperforms the worst-case tuning $\rho_\frac{X}{R}=0$, i.e., achieves higher VSC terminal voltage with lower current magnitude. Notably, compared to the variable virtual impedance with optimal tuning, constraint-aware GFM control reaches the maximum current and provides improved voltage support. To demonstrate that constraint-aware GFM control can match the voltage support of variable virtual impedance with optimal tuning without knowledge of the grid reactance-to-resistance ratio, the current limit of constraint-aware GFM control is adjusted to match the current magnitude of variable virtual impedance. It can be seen that constraint-aware GFM control closely matches the voltage support of the variable virtual impedance without knowledge of the grid reactance-to-resistance ratio.

\begin{figure}[b!]
  \centering
\includegraphics[trim={0.55cm 0 0 0},width=0.85\columnwidth,clip]{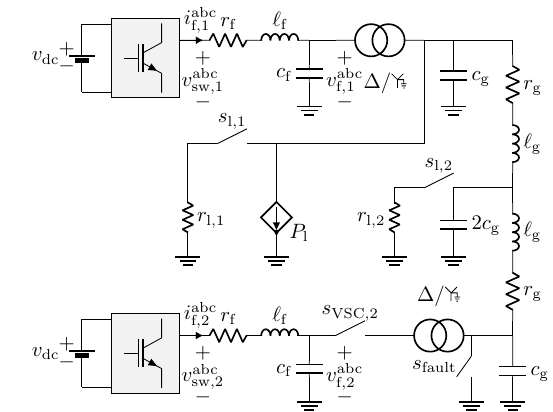}
\caption{Two bus system consisting of two $1$~MW converters, $\Delta\wyegnd$ step up transformers, two $1$~km distribution lines with inductance $\ell_\text{g}$, resistance $r_\text{g}$, and shunt capacitance $\ell_\text{g}$, a $500$~kW resistive load ($r_{\text{l},1}$), a $250$~kW resistive load ($r_{\text{l},1}$), and a time-varying constant power load $P_\text{l}$.\label{fig:twobus}}
\end{figure}
\begin{figure*}[t]
  \centering
\includegraphics[width=1\textwidth,clip]{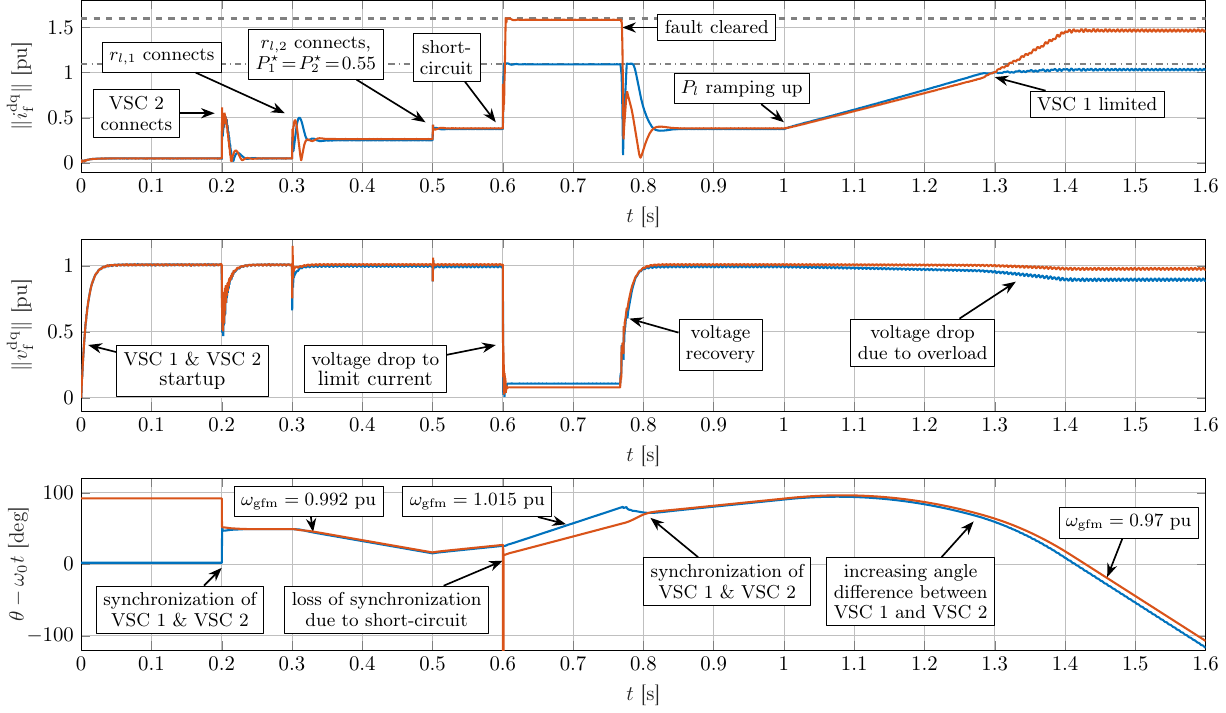}
\caption{Simulation results for the two bus system: VSC~1 (\protect\blueline), VSC~2 (\protect\redline), $i_{\max,1}$ (\protect\greydd), $i_{\max,2}$ (\protect\greyd).\label{fig:twobus:simres}}
\end{figure*}

\section{Interactions Between Converters}\label{sec:twobus}
Next, we simulate the two-bus system shown in Fig.~\ref{fig:twobus} to illustrate interactions of converters during synchronization, fault recovery, and overload. This simulation {uses two-level VSCs energized by a dc voltage source, IGBT switch model, sinusoidal pulse width modulation (PWM) with $10$~kHz switching frequency,} $3~\mu$s controller computation delay, and simulation step size $\tau_s=0.1~\mu$s. The parameters for this case study are provided in Tab.~\ref{tab:twobus}.
\begin{table}[b!]
\begin{center}
\caption{Parameters for the two bus system\label{tab:twobus}}
\begin{tabular}{ cccccc }          \toprule 
    \multicolumn{2}{c}{Grid base values} & \multicolumn{2}{c}{Filter (in VSC base)} & \multicolumn{2}{c}{VSC base values} \\ \toprule 
$V_\text{base}$ & $4.16$~kV &  $\ell_\text{f}$ & $0.1$~pu & $V_{\text{base},\text{vsc}}$ & $480$~V \\
$S_\text{base}$ & $1.5$~MW &  $r_f$ & $0.01$~pu & $S_{\text{base},\text{vsc}}$ &  $1$~MW \\ 
$f_\text{base}$ & $60$~Hz & $c_\text{f}$ & $0.05$~pu & $v_\text{dc}$ & $1000$~V  \\[0.2em] \toprule 
  \multicolumn{6}{c}{Control} \\ \toprule 
  $m_p$, $m_q$ & $3$\% & $\tau_\text{lp}$ & $5.3$~ms & $k_\text{rc}$ & $0.1$~pu\\
  $w^\text{pu}_\omega$ & $0.1$ & $\tau_\text{ctr}$ & $0.1$~ms & $\omega_\text{rc}$ & $10^4$~rad/s \\ 
  $\tau_\text{v}$ &  $8$~ms &  $\tau_\text{cyc}$ & $16.7$~ms & $f_\text{PWM}$ & $10$~kHz \\[0.2em] \toprule 
  \multicolumn{2}{c}{Constraints} & \multicolumn{2}{c}{distribution line (per km)}  & \multicolumn{2}{c}{Transformer}\\ \toprule
  $V_{\max}$ & $1.276$~pu &   $\ell_\text{g}$ & $\!5.56\!\times\!10^{-2}\!$~pu & $\ell_{\Delta\wyegnd}$ & $0.03$~pu \\
   $i_{\max,1}$ & $1.1$~pu  & $r_\text{g}$ & $\!1.82\!\times\!10^{-2}\!$~pu & $r_{\Delta\wyegnd}$ & $0.002$~pu  \\
   $i_{\max,2}$ & $1.6$~pu & $c_\text{g}$ & $\!4.34 \!\times\! 10^{-5}\!$~pu \\   \bottomrule
\end{tabular}
\end{center}
\end{table}
To facilitate a trade-off between stabilizing voltage magnitude and frequency (see Sec.~\ref{sec:gfm:disc}) we select $w_\theta = \frac{w^\text{pu}_{\omega}}{\omega_0 \tau_\text{ctr}}$ with $w^\text{pu}_\omega=0.1$. A larger penalty $w^\text{pu}_\omega$ results in closer tracking of the GFM frequency \eqref{eq:droop:cont:freq} but a potential loss of stability during scenarios that require significant voltage support. 

Simulation results for current magnitude, voltage magnitude, and voltage phase angle are shown in Fig.~\ref{fig:twobus:simres}. The system is first energized by VSC~1 and VSC~2 starts up disconnected from the system. Next, at $t=0.2$~s the breaker $s_{\text{VSC},2}$ closes to connect VSC~2 to the system. Subsequently, the breaker $s_{l,1}$ and $s_{l,2}$ close at $t=0.3$~s and $t=0.5$~s respectively to connect $500$~kW and $250$~kW resistive load. Moreover, at $t=0.5$~s the active power set points of VSC~1 and VSC~2 are updated to $0.55$~pu. Finally, from $t=0.6$~s to $t=0.76$~s a short-circuit fault is simulated by closing $s_\text{fault}$. Finally, from $t=1$~s to $t=1.4$~s the constant power load $P_l$ ramps from $0$~kW to $800$~kW. We emphasize that no additional startup or synchronization control is used. 
\begin{figure}[h!]
  \centering
\includegraphics[width=0.99\columnwidth,clip]{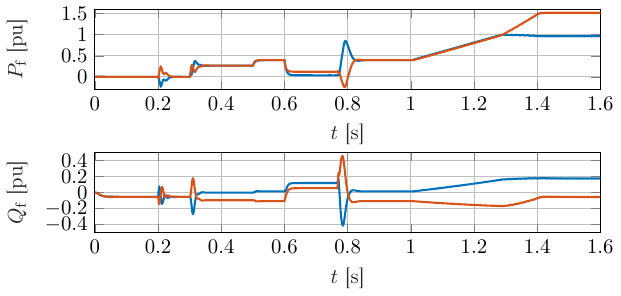}
\caption{Active and reactive power injection of VSC~1 (\protect\blueline) and VSC~2 (\protect\redline).\label{fig:twobus:simres:PQ}}
\end{figure} 

It can be seen that the VSCs operate within their limits and rapidly synchronize with each other upon connection of VSC~2 and fault clearing. Moreover, after VSC~1 becomes current limited VSC~2 picks up the subsequent increase of the constant power load $P_l$ while VSC~1 continues to provide  voltage support by injecting reactive power and maintains a few percent of headroom to retain some GFM features (e.g., voltage support). After VSC~2 becomes limited, the frequency drops to $0.97$~pu and significantly deviates from the design assumption $\omega_\text{dq}  =\omega_0$ of \eqref{eq:ADMM}. The active and reactive power injection of VSC~1 and VSC~2 are shown in Fig.~\ref{fig:twobus:simres:PQ}.

Increasing the total load beyond {$1.55$~MW} results in insufficient headroom to maintain voltage and frequency stability as both VSCs become overloaded. Using variable virtual impedance or current reference limiting for both VSCs with the parameters used in Sec.~\ref{sec:comp}, increasing the load beyond $1.31$~MW (variable virtual impedance) and $1.41$~MW (current reference limiting) results in voltage collapse. In other words, constraint-aware GFM control can accommodate $10\%$ to $20\%$ higher load in this case study than commonly used current limiters.

Finally, current and voltage waveforms for VSC~2 during connection and fault inception are shown in Fig.~\ref{fig:twobus:simres:wave}. It can be seen that voltages quickly settle to post-event steady-states and filter resonance is well controlled and $LC(L)$ filter resonance (see Sec.~\ref{subsec:symshort}) is suppressed within half a cycle after connecting VSC~2 and fault inception.
\begin{figure}[h!]
  \centering
\includegraphics[width=0.99\columnwidth,clip]{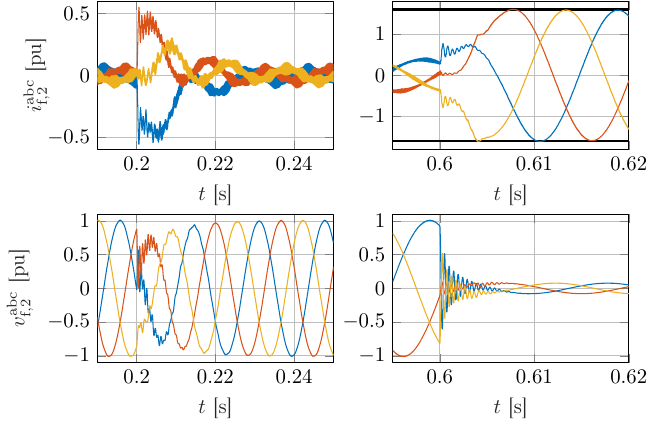}
\caption{Phase currents and voltages for VSC~2 during connection and fault inception: phase $a$ (\protect\blueline), phase $b$ (\protect\redline), phase $c$ (\protect\yellowline), and current limits (\protect\blackline). \label{fig:twobus:simres:wave}}
\end{figure}

\section{Hardware experiments}\label{sec:experiment}
{\subsection{Description of the Hardware Testbed and Experiments}}
To validate the results, the proposed control has been implemented on an experimental testbed (see Fig.~\ref{fig:testbed:picture}) consisting of a VSC, grid impedance $\ell_\text{g}$, and voltage source emulating the grid voltage $v^\text{abc}_\text{g}$ (see Fig.~\ref{fig:testbed:schematic} and Tab.~\ref{tab:testbed}).
\begin{figure}[h!]
  \centering
  \includegraphics[width=0.99\columnwidth]{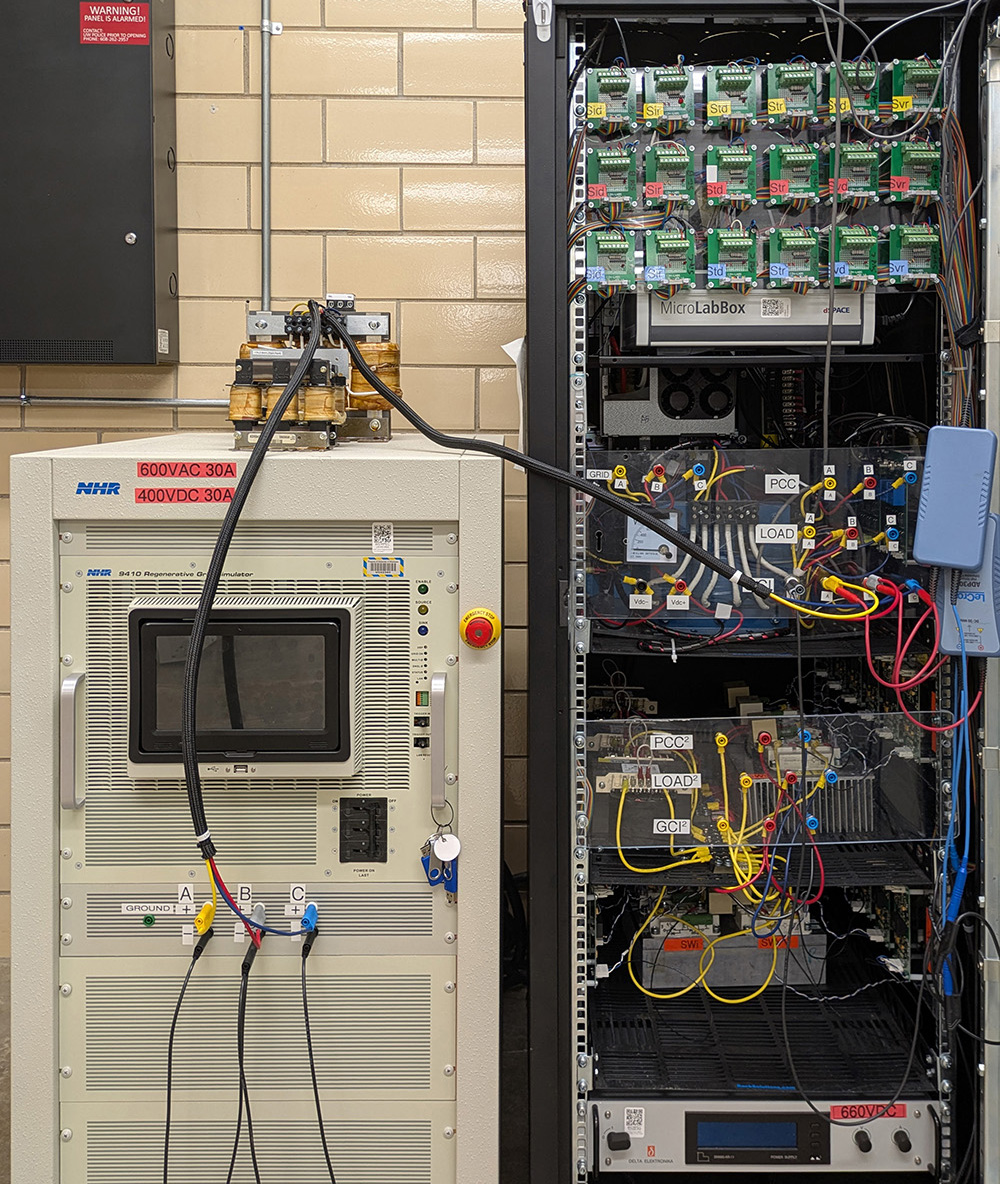}
\caption{Experimental testbed with NHR-9420 bidirectional ac source, inductors used for grid impedance emulation, voltage source converter, and dSPACE MicroLabBox.\label{fig:testbed:picture}}
\end{figure}
\begin{figure}[b!]
  \centering
\includegraphics[trim={0.55cm 0 0 0},width=0.75\columnwidth,clip]{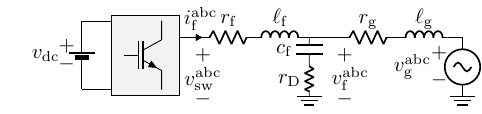}
\caption{Experimental setup consisting of a two-level VSC, grid impedance $\ell_\text{g}$, $r_\text{g}$, and grid emulator $v^\text{abc}_\text{g}$.\label{fig:testbed:schematic}}
\end{figure}

\begin{figure*} 
  \centering
\includegraphics[width=1\textwidth]{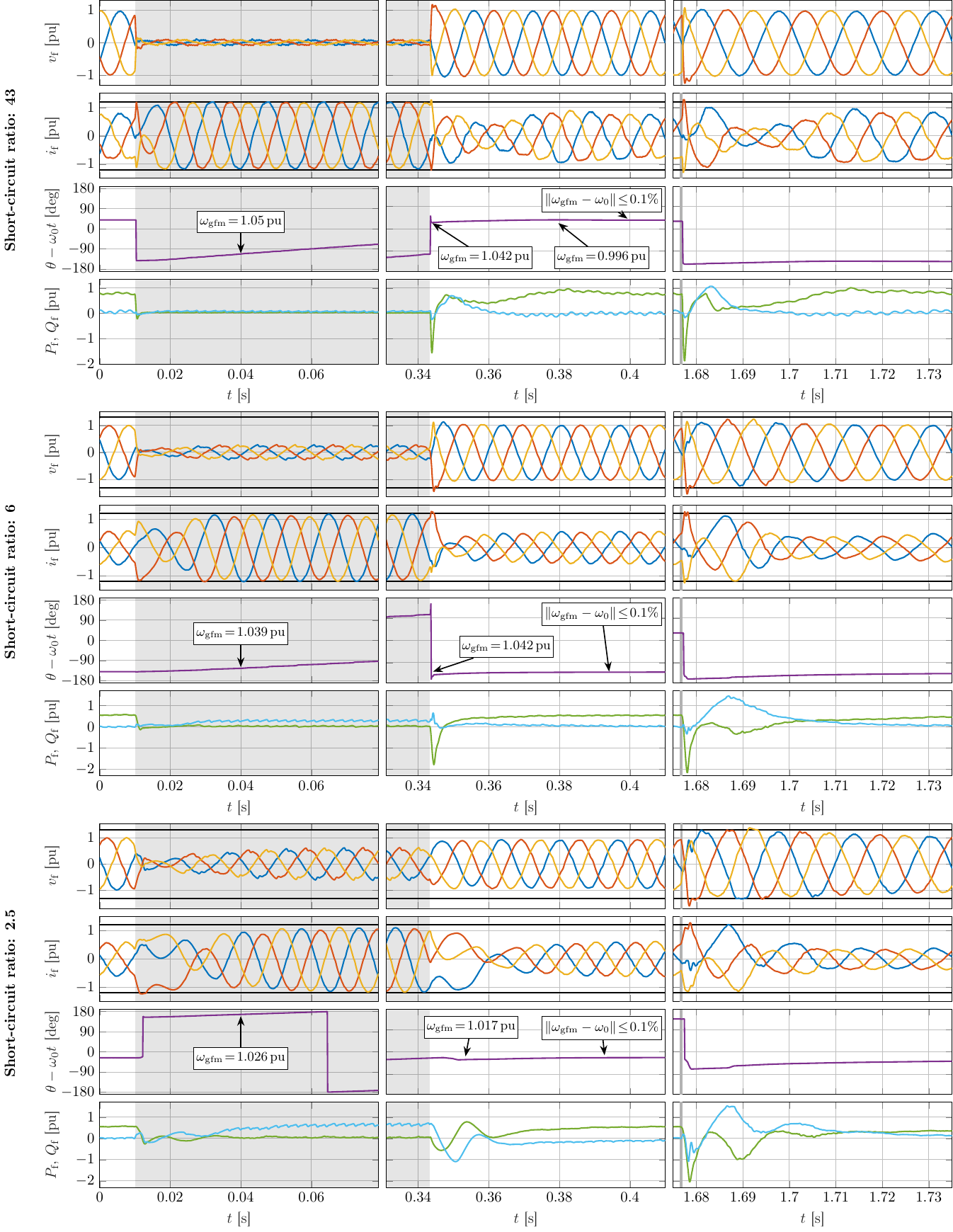}
\vspace{-2em}
\caption{Experimental results for a symmetric fault from $t\!=\!0.01$~s to $t\!=\!0.3433$~s (gray) and $180^\circ$ phase jump at $t\!=\!1.6767$~s (gray): filter voltages and currents (phase $a$ (\protect\blueline), phase $b$ (\protect\redline), phase $c$ (\protect\yellowline)), current and modulation limits (\protect\blackline), GFM control phase angle (\protect\purpleline), active  (\protect\greenline)  and reactive  (\protect\cyanline) power.\label{fig:exp:all}}
\end{figure*}
\begin{table}[b!]
\begin{center}
\caption{Parameters for the experimental setup\label{tab:testbed}}
\begin{tabular}{ cccccc }          \toprule 
  \multicolumn{2}{c}{Grid} & \multicolumn{4}{c}{VSC} \\ \toprule 
$V_\text{base}$ & $186$~V & $v_\text{dc}$ & $400$~V & $r_\text{D}$ & $0.31$~pu  \\
$S_\text{base}$ & $1.788$~kW & $\ell_\text{f}$ & $0.084$~pu & $V_{\max}$ & $1.31$~pu  \\ 
$f_\text{base}$ & $60$~Hz & $c_\text{f}$  & $0.08$~pu & $i_{\max}$  & $1.2$~pu  \\[0.2em] \toprule 
  \multicolumn{6}{c}{GFM control} \\ \toprule
  $m_p$ & $5$\% & $\tau_\text{lp}$ & $5.3$~ms  & $\tau_\text{ctr}$ & $0.1$~ms\\
  $m_q$ & $5$\% & $w^\text{pu}_\omega$ & $0.2$  & $\tau_\text{cyc}$ & $20$~ms \\  \bottomrule
\end{tabular}
\end{center}
\end{table}
A $10$~kW two-level VSC previously used in~\cite{CJS2020} has been derated for safety reasons and the $LC$ filter has been modified to match the new rating of $1.788$~kW. A NHR-9420 bidirectional ac source is used to emulate the grid voltage $v^\text{abc}_\text{g}$. A Delta Elektronika SM660-AR-11 unidirectional dc source is used to emulate the converter dc bus and, in abstraction, model energy storage.

Constraint-aware GFM control {shown in Fig.~\ref{fig:proj.gfm} with projection algorithm shown in Fig.~\ref{fig:ADMM}} is implemented {on the NXP QoriQ P5020 microprocessor of a dSPACE MicroLabBox using Simulink Coder}. Sinusoidal PWM with switching frequency $f_\text{sw}=10$~kHz and synchronous sampling with controller update at the PWM midpoint are used. The ADMM algorithm \eqref{eq:ADMM} with $\rho=5$, $n_\text{it}=5$, and $\alpha=1.6$ is used to compute the projection onto the constraint set. {The overall computation time is $12$~$\mu$s on a single core of the microprocessor. Notably, only $2$~$\mu$s are spent on evaluating the GFM control, while $10$~$\mu$s are spent on I/O tasks, filtering measurements, and computing duty cycles. An in-depth discussion of the computational complexity of the proposed constraint-aware GFM control can be found in Sec.~\ref{subsec:complexity}.}

Finally, various off-the-shelf three-phase line reactors (TCI KDRB43H, TCI KDRB43H, TCI KDRH1FE, and MTE RL-04503) with reactance-to-resistance ratio $\omega_0 \ell_\text{g} / r_\text{g} \in [20,40]$ are used to replicate a grid SCR of approximately $2.5$ ($\ell_\text{g}=21.2$~mH), $6$ ($\ell_\text{g}=8.6$~mH), and $43$ ($\ell_\text{g}=1.2$~mH). The experimental testbed is shown in Fig.~\ref{fig:testbed:picture}. {A symmetric short-circuit fault is emulated} by setting $V_\text{g}=0$. A $180^\circ$ phase jump of the grid voltage $v^\text{abc}_\text{g}$ is used to demonstrate the ability of the control to synchronize from arbitrary initial angles without significant violation of the current limit. Signals captured at the PWM midpoint {for a wide range of short-circuit ratios} are shown in Fig.~\ref{fig:exp:all}.

\subsection{{Very High Short-Circuit Ratio}}
While GFL converters exhibit challenges under weak grid connection (i.e., low SCR), GFM converters typically exhibit various challenges under strong grid connection (i.e., high SCR)~\cite{GCB+2019,LGG2022}. To verify that the proposed control successfully limits the current even under extreme scenarios, an SCR significantly exceeding the typical range (i.e., $2$ to $20$) and pre-contingency operating point of $P^\star=0.8$ and $Q^\star=0$ has been selected. The results for a short-circuit ratio of $43$ and {$\tau_\text{v}=8$~ms} show rapid settling of the voltage to post-contingency steady-state during fault inception, fault clearing, and after the $180$-degree phase jump. While the converter and grid voltage exhibit no significant distortion, the pre- and post-fault current and power exhibit distortions that can be attributed to tolerance of phase inductances relative to the low impedance between the two parallel voltage sources. This effect is unrelated to the proposed limiter. 

While brief and small violations of the current limit occur immediately after the phase jump, the current remains well within typical short-term and long-term VSC current limits. During the short-circuit fault, the average frequency (i.e., derivative of $\theta(t)$ in Fig.~\ref{fig:exp:all}) is approximately $1.05$~pu, i.e., slightly higher than the $1.04$~pu expected for $P^\star=0.8$ and $m_p=0.05$. The limiter deactivates $0.5$~ms after fault clearing and a two to three cycle resynchronization transient of the unconstrained droop control can be observed. Due to the strong coupling, this transient significantly impacts the current. Finally, we note that harmonics can be observed in the voltage that can be attributed to the converter switching.

\subsection{Medium and Low Short-Circuit Ratio}
For the remaining experiments we use a pre-contingency operating point of $P^\star=0.55$ and $Q^\star=0$. Moreover, for {lower} short-circuit ratios the physics of the {grid} necessitate slowing down the voltage recovery to avoid oscillations. To account for this aspect, $\tau_\text{v}=8$~ms is used for short-circuit ratios above $5$ while $\tau_\text{v}=16$~ms is used for short-circuit ratios below $5$. The results for medium and low short-circuit ratios exhibit increased settling times of voltages and currents due to increased time constants of the grid current. The VSC currents again only exhibit small and brief violations of the current limit that are well within typical short-term and long-term VSC current limits. While the voltage is well controlled during the short-circuit fault, some harmonics are visible that may be attributed to the fact that the inductors used for grid impedance emulation are not perfectly phase-balanced. The average frequency for both cases is approximately $1.025$~pu, i.e., as expected for $P^\star\!=\!0.5$ and $m_p\!=\!0.05$.

\subsection{Unbalanced Fault}
Finally, the response to a phase $a$ to ground fault is tested using the setup shown in Fig.~\ref{fig:testbed:unbal}. Experimental results are shown in Fig.~\ref{fig:exp:unbal}. Despite only being designed for balanced systems, the proposed control successfully limits the current. However, the resulting waveforms are distorted because the control design does not account for negative sequence components that result in non-constant steady-state components in the positive sequence dq-frame of the controller.
\begin{figure}[h!]
  \centering
\includegraphics[trim={0.55cm 0 0 0},width=0.85\columnwidth,clip]{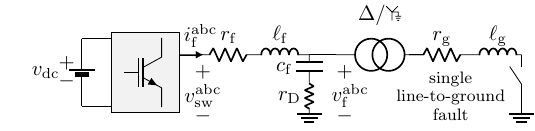}
\caption{Two-level VSC, $\Wyegnd / \Delta$ transformer with $2.6$~\% impedance, and $7.2$~mH grid impedance.\label{fig:testbed:unbal}}
\end{figure}
\begin{figure}[h]
  \centering
\includegraphics[width=0.99\columnwidth]{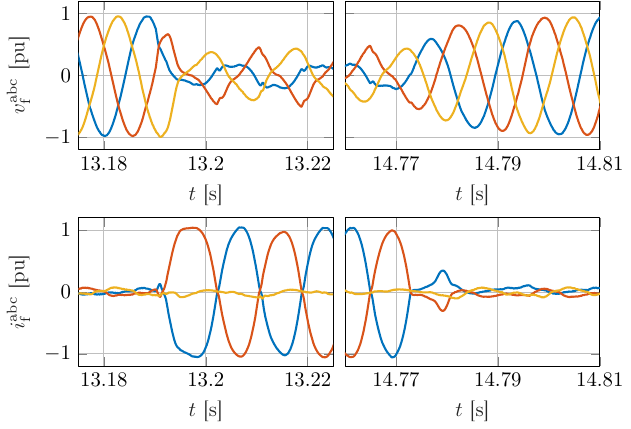}
\caption{VSC filter voltages and currents (phase $a$ (\protect\blueline), phase $b$ (\protect\redline), phase $c$ (\protect\yellowline)) during a phase $a$ to ground fault in a grid connected via a $\Wyegnd / \Delta$ transformer.\label{fig:exp:unbal}}
\end{figure}

\section{Conclusion and Outlook}\label{sec:concl}
In this work, a constraint-aware GFM control is developed that explicitly integrates constraints (i.e., current and modulation limits) into GFM droop control. To enable a systematic control design, the set of feasible converter voltages (i.e., that do not result in constraint violations) is characterized and analyzed. Moreover, an efficient algorithm for projecting voltages onto the feasible set is developed. Subsequently, these preliminary results are combined with GFM droop control to obtain a constraint-aware GFM control that minimizes the deviation from the well-known unconstrained GFM droop response under modulation and current limits. The resulting constraint-aware GFM control enables fast current limiting, enhanced transient stability, and infinite critical clearing time. Simulation studies have been used to illustrate and validate the constraint-aware GFM control for single-converter infinite-bus and two-converter test cases. Finally, hardware experiments have been conducted to validate the response to short-circuit faults and phase jumps. While the results are encouraging, this work opens numerous further research directions. {In particular, extensions to unbalanced grid conditions and faults are required to make the results practical in transmission system applications. Moreover, considering resource dynamics (e.g., PV, wind turbines) and constraints is seen as an important topic for future work.}

\bibliographystyle{IEEEtran}
\bibliography{IEEEabrv,bib_file}

\begin{thebibliography}{10}
\providecommand{\url}[1]{#1}
\csname url@samestyle\endcsname
\providecommand{\newblock}{\relax}
\providecommand{\bibinfo}[2]{#2}
\providecommand{\BIBentrySTDinterwordspacing}{\spaceskip=0pt\relax}
\providecommand{\BIBentryALTinterwordstretchfactor}{4}
\providecommand{\BIBentryALTinterwordspacing}{\spaceskip=\fontdimen2\font plus
\BIBentryALTinterwordstretchfactor\fontdimen3\font minus
  \fontdimen4\font\relax}
\providecommand{\BIBforeignlanguage}[2]{{%
\expandafter\ifx\csname l@#1\endcsname\relax
\typeout{** WARNING: IEEEtran.bst: No hyphenation pattern has been}%
\typeout{** loaded for the language `#1'. Using the pattern for}%
\typeout{** the default language instead.}%
\else
\language=\csname l@#1\endcsname
\fi
#2}}
\providecommand{\BIBdecl}{\relax}
\BIBdecl

\bibitem{GD2022}
F.~D\"orfler and D.~Gro\ss, ``Control of low-inertia power systems,''
  \emph{Annual Review of Control, Robotics, and Autonomous Systems}, vol.~6,
  no.~1, pp. 415--445, 2023.

\bibitem{MDH+18}
F.~Milano, F.~Dörfler, G.~Hug, D.~J. Hill, and G.~Verbič, ``Foundations and
  challenges of low-inertia systems (invited paper),'' in \emph{Power Systems
  Computation Conference}, 2018.

\bibitem{CDA93}
M.~Chandorkar, D.~Divan, and R.~Adapa, ``Control of parallel connected
  inverters in standalone ac supply systems,'' \emph{{IEEE} Trans. Ind. Appl.},
  vol.~29, no.~1, pp. 136--143, 1993.

\bibitem{DSF2015}
S.~D’Arco, J.~A. Suul, and O.~B. Fosso, ``A virtual synchronous machine
  implementation for distributed control of power converters in smartgrids,''
  \emph{Electr. Pow. Sys. Res.}, vol. 122, pp. 180--197, 2015.

\bibitem{JD+2014}
B.~B. Johnson, S.~V. Dhople, A.~O. Hamadeh, and P.~T. Krein, ``Synchronization
  of parallel single-phase inverters with virtual oscillator control,''
  \emph{{IEEE} Trans. Power Electron.}, vol.~29, no.~11, pp. 6124--6138, 2014.

\bibitem{GCB+2019}
D.~Groß, M.~Colombino, J.-S. Brouillon, and F.~Dörfler, ``The effect of
  transmission-line dynamics on grid-forming dispatchable virtual oscillator
  control,'' \emph{{IEEE} Trans. Control Netw. Syst.}, vol.~6, no.~3, pp.
  1148--1160, 2019.

\bibitem{JY+2018}
J.~Jia, G.~Yang, and A.~H. Nielsen, ``A review on grid-connected converter
  control for short-circuit power provision under grid unbalanced faults,''
  \emph{{IEEE} Trans. Power Del.}, vol.~33, no.~2, pp. 649--661, 2018.

\bibitem{DPD+2018}
G.~Denis, T.~Prevost, M.-S. Debry, F.~Xavier, X.~Guillaud, and A.~Menze, ``The
  {M}igrate project: the challenges of operating a transmission grid with only
  inverter-based generation. a grid-forming control improvement with transient
  current-limiting control,'' \emph{IET Renewable Power Generation}, vol.~12,
  no.~5, pp. 523--529, 2018.

\bibitem{BCL+2024}
N.~Baeckeland, D.~Chatterjee, M.~Lu, B.~Johnson, and G.-S. Seo, ``Overcurrent
  limiting in grid-forming inverters: A comprehensive review and discussion,''
  \emph{{IEEE} Trans. Power Electron.}, vol.~39, no.~11, pp. 14\,493--14\,517,
  2024.

\bibitem{XHZ+2016}
H.~Xin, L.~Huang, L.~Zhang, Z.~Wang, and J.~Hu, ``Synchronous instability
  mechanism of {P-f} droop-controlled voltage source converter caused by
  current saturation,'' \emph{{IEEE} Trans. Power Syst.}, vol.~31, no.~6, pp.
  5206--5207, 2016.

\bibitem{QGC+18}
T.~Qoria, F.~Gruson, F.~Colas, X.~Guillaud, M.~Debry, and T.~Prevost, ``Tuning
  of cascaded controllers for robust grid-forming voltage source converter,''
  in \emph{Power Systems Computation Conference}, 2018.

\bibitem{GG2023}
Y.~Gu and T.~C. Green, ``Power system stability with a high penetration of
  inverter-based resources,'' \emph{Proceedings of the IEEE}, vol. 111, no.~7,
  pp. 832--853, 2023.

\bibitem{FCG2010}
F.~Salha, F.~Colas, and X.~Guillaud, ``Virtual resistance principle for the
  overcurrent protection of pwm voltage source inverter,'' in \emph{IEEE PES
  Innovative Smart Grid Technologies Conference Europe}, 2010.

\bibitem{HL2011}
J.~He and Y.~W. Li, ``Analysis, design, and implementation of virtual impedance
  for power electronics interfaced distributed generation,'' \emph{{IEEE}
  Trans. Ind. Appl.}, vol.~47, no.~6, pp. 2525--2538, 2011.

\bibitem{QWW+2023}
T.~Qoria, H.~Wu, X.~Wang, and I.~Colak, ``Variable virtual impedance-based
  overcurrent protection for grid-forming inverters: Small-signal, large-signal
  analysis and improvement,'' \emph{{IEEE} Trans. Smart Grid}, vol.~14, no.~5,
  pp. 3324--3336, 2023.

\bibitem{ZBS+2023}
Z.~Zeng, P.~Bhagwat, M.~Saeedifard, and D.~Groß, ``Hybrid threshold virtual
  impedance for fault current limiting in grid-forming converters,'' in
  \emph{IEEE Energy Convers. Congr. Expo.}, 2023, pp. 913--918.

\bibitem{AA2023}
A.~Acharya and R.~Ayyanar, ``Enhancing stability of {dVOC} controlled
  grid-forming inverters under large grid transients — a power angle based
  approach,'' in \emph{IEEE Energy Convers. Congr. Expo.}, 2023, pp. 803--808.

\bibitem{GL2023}
D.~Groß and X.~Lyu, ``Towards constrained grid-forming control,'' in
  \emph{Allerton Conference on Communication, Control, and Computing}, 2023.

\bibitem{GD2019}
D.~Groß and F.~Dörfler, ``Projected grid-forming control for current-limiting
  of power converters,'' in \emph{Allerton Conference on Communication,
  Control, and Computing}, 2019, pp. 326--333.

\bibitem{EUJ2022}
T.~Erckrath, P.~Unruh, and M.~Jung, ``Voltage phasor based current limiting for
  grid-forming converters,'' in \emph{IEEE Energy Convers. Congr. Expo.}, 2022.

\bibitem{CPO2020}
J.~Chen, F.~Prystupczuk, and T.~O'Donnell, ``Use of voltage limits for current
  limitations in grid-forming converters,'' \emph{CSEE Journal of Power and
  Energy Systems}, vol.~6, no.~2, pp. 259--269, 2020.

\bibitem{ALJ+2022}
O.~Ajala, M.~Lu, B.~Johnson, S.~V. Dhople, and A.~Domínguez-García, ``Model
  reduction for inverters with current limiting and dispatchable virtual
  oscillator control,'' \emph{{IEEE} Trans. Energy Convers.}, vol.~37, no.~4,
  pp. 2250--2259, 2022.

\bibitem{ARY+2023}
M.~A. Awal, M.~R.~K. Rachi, H.~Yu, S.~Schr\"oder, J.~Dannehl, and I.~Husain,
  ``Grid-forming nature retaining fault ride-through control,'' in
  \emph{Applied Power Electronics Conf. and Expo.}, 2023, pp. 2753--2758.

\bibitem{HDH+2025}
X.~He, M.~A. Desai, L.~Huang, and F.~Dörfler, ``Cross-forming control and
  fault current limiting for grid-forming inverters,'' \emph{{IEEE} Trans.
  Power Electron.}, vol.~40, no.~3, pp. 3980--4007, 2025.

\bibitem{AGP+2024}
A.~Arjomandi-Nezhad, Y.~Guo, B.~C. Pal, and D.~Varagnolo, ``A model predictive
  approach for enhancing transient stability of grid-forming converters,''
  \emph{{IEEE} Trans. Power Syst.}, vol.~39, no.~5, pp. 6675--6688, 2024.

\bibitem{CKK+2008}
P.~Cortes, M.~P. Kazmierkowski, R.~M. Kennel, D.~E. Quevedo, and J.~Rodriguez,
  ``Predictive control in power electronics and drives,'' \emph{{IEEE} Trans.
  Ind. Electron.}, vol.~55, no.~12, pp. 4312--4324, 2008.

\bibitem{QAP+2012}
D.~E. Quevedo, R.~P. Aguilera, M.~A. Perez, P.~Cortes, and R.~Lizana, ``Model
  predictive control of an {AFE} rectifier with dynamic references,''
  \emph{{IEEE} Trans. Power Electron.}, vol.~27, no.~7, pp. 3128--3136, 2012.

\bibitem{KG2020}
P.~Karamanakos and T.~Geyer, ``Guidelines for the design of finite control set
  model predictive controllers,'' \emph{{IEEE} Trans. Power Electron.},
  vol.~35, no.~7, pp. 7434--7450, 2020.

\bibitem{RKC+2022}
M.~Rossi, P.~Karamanakos, and F.~Castelli-Dezza, ``An indirect model predictive
  control method for grid-connected three-level neutral point clamped
  converters with {$LCL$} filters,'' \emph{{IEEE} Trans. Ind. Appl.}, vol.~58,
  no.~3, pp. 3750--3768, 2022.

\bibitem{EAC+1992}
{IEEE Power System Relaying Committee Working Group}, ``Single phase tripping
  and auto reclosing of transmission lines-ieee committee report,''
  \emph{{IEEE} Trans. Power Del.}, vol.~7, no.~1, pp. 182--192, 1992.

\bibitem{REL2021}
R.~Rosso, S.~Engelken, and M.~Liserre, ``On the implementation of an {FRT}
  strategy for grid-forming converters under symmetrical and asymmetrical grid
  faults,'' \emph{{IEEE} Trans. Ind. Appl.}, vol.~57, no.~5, pp. 4385--4397,
  2021.

\bibitem{BVKS+2022}
N.~Baeckeland, D.~Venkatramanan, M.~Kleemann, and S.~Dhople, ``Stationary-frame
  grid-forming inverter control architectures for unbalanced fault-current
  limiting,'' \emph{{IEEE} Trans. Energy Convers.}, vol.~37, no.~4, pp.
  2813--2825, 2022.

\bibitem{BG2023}
P.~Bhagwat and D.~Gro{\ss}, ``Three-phase grid-forming droop control for
  unbalanced systems and fault ride through,'' in \emph{IEEE Power \& Energy
  Society General Meeting}, 2023.

\bibitem{HNL+2015}
P.~J. Hart, A.~Nelson, R.~H. Lasseter, and T.~M. Jahns, ``Effect of power
  measurement filter properties on {CERTS} microgrid control performance,'' in
  \emph{IEEE Int. Symposium on Power Electronics for Distributed Generation
  Systems}, 2015.

\bibitem{WWZ2024}
H.~Wu, X.~Wang, and L.~Zhao, ``Design considerations of current-limiting
  control for grid-forming capability enhancement of vscs under large grid
  disturbances,'' \emph{{IEEE} Trans. Power Electron.}, vol.~39, no.~10, pp.
  12\,081--12\,085, 2024.

\bibitem{MSB+2018}
H.~Matsumori, T.~Shimizu, F.~Blaabjerg, X.~Wang, and D.~Yang, ``Stability
  influence of filter components parasitic resistance on {LCL}-filtered grid
  converters,'' in \emph{International Power Electronics Conference}, 2018, pp.
  3357--3362.

\bibitem{PCM1980}
W.~M. Polivka, P.~R. Chetty, and R.~D. Middlebrook, ``State-space average
  modelling of converters with parasitics and storage-time modulation,'' in
  \emph{IEEE Power Electronics Specialists Conference}, 1980, pp. 119--143.

\bibitem{LWL+2018}
M.~Lu, X.~Wang, P.~C. Loh, F.~Blaabjerg, and T.~Dragicevic, ``Graphical
  evaluation of time-delay compensation techniques for digitally controlled
  converters,'' \emph{IEEE Transactions on Power Electronics}, vol.~33, no.~3,
  pp. 2601--2614, 2018.

\bibitem{MAL-016}
S.~Boyd, N.~Parikh, E.~Chu, B.~Peleato, and J.~Eckstein, ``Distributed
  optimization and statistical learning via the alternating direction method of
  multipliers,'' \emph{Foundations and Trends{\textregistered} in Machine
  Learning}, vol.~3, no.~1, pp. 1--122, 2011.

\bibitem{osqp}
B.~Stellato, G.~Banjac, P.~Goulart, A.~Bemporad, and S.~Boyd, ``{OSQP}: an
  operator splitting solver for quadratic programs,'' \emph{Mathematical
  Programming Computation}, vol.~12, no.~4, pp. 637--672, 2020.

\bibitem{WBL+2015}
X.~Wang, F.~Blaabjerg, and P.~C. Loh, ``Virtual {RC} damping of {LCL}-filtered
  voltage source converters with extended selective harmonic compensation,''
  \emph{{IEEE} Trans. Power Electron.}, vol.~30, no.~9, pp. 4726--4737, 2015.

\bibitem{WLB+2015}
X.~Wang, Y.~W. Li, F.~Blaabjerg, and P.~C. Loh, ``Virtual-impedance-based
  control for voltage-source and current-source converters,'' \emph{{IEEE}
  Trans. Power Electron.}, vol.~30, no.~12, pp. 7019--7037, 2015.

\bibitem{SGR+2013}
J.~Schiffer, D.~Goldin, J.~Raisch, and T.~Sezi, ``Synchronization of
  droop-controlled microgrids with distributed rotational and electronic
  generation,'' in \emph{IEEE Conf. on Dec. and Contr.}, 2013, pp. 2334--2339.

\bibitem{FW2022}
B.~Fan and X.~Wang, ``Equivalent circuit model of grid-forming converters with
  circular current limiter for transient stability analysis,'' \emph{{IEEE}
  Trans. Power Syst.}, vol.~37, no.~4, pp. 3141--3144, 2022.

\bibitem{TI_TMS320F2837xD_SPRS880P_2024}
\emph{TMS320F2837xD Dual-Core Real-Time Microcontrollers datasheet}, Texas
  Instruments, Rev. P, Feb. 2024, sPRS880P.

\bibitem{Freescale_E600CORERM_2006}
\emph{{e600} PowerPC\texttrademark{} Core Reference Manual}, Freescale
  Semiconductor, Inc., Mar. 2006, e600CORERM, Rev. 0.

\bibitem{GP2016}
L.~Gr{\"u}ne and J.~Pannek, \emph{Nonlinear Model Predictive Control: Theory
  and Algorithms}.\hskip 1em plus 0.5em minus 0.4em\relax Springer
  International Publishing, 2016.

\bibitem{UNIFI_Specs_V2_2024}
``{UNIFI} specifications for grid-forming inverter-based resources — version
  2,'' {UNIFI Consortium}, Tech. Rep., Mar. 2024, {UNIFI-2024-2-1}.

\bibitem{ENTSOE_GFM_PPM_Interim_2024}
``Grid forming capability of power park modules: First interim report on
  technical requirements,'' ENTSO-E, Brussels, Belgium, Interim report, May
  2024.

\bibitem{SG+20}
I.~Suboti\`c, D.~Gro\ss, M.~Colombino, and F.~Dörfler, ``A {L}yapunov
  framework for nested dynamical systems on multiple time scales with
  application to converter-based power systems,'' \emph{{IEEE} Trans. Autom.
  Control}, vol.~66, no.~12, pp. 5909--5924, 2021.

\bibitem{IEEE2800-2022}
``{IEEE} standard for interconnection and interoperability of inverter-based
  resources {(IBRs}) interconnecting with associated transmission electric
  power systems,'' \emph{IEEE Std 2800-2022}, 2022.

\bibitem{CJS2020}
W.~Choi, K.~Jung, and B.~Sarlioglu, ``Power control of hybrid grid-connected
  inverter to improve power quality,'' in \emph{IEEE Energy Convers. Congr.
  Expo.}, 2020, pp. 3741--3745.

\bibitem{LGG2022}
Y.~Li, Y.~Gu, and T.~C. Green, ``Revisiting grid-forming and grid-following
  inverters: A duality theory,'' \emph{{IEEE} Trans. Power Syst.}, vol.~37,
  no.~6, pp. 4541--4554, 2022.

\end{thebibliography}

\end{document}